\pdfoutput=1
\documentclass[11pt]{article}
\usepackage[margin=1in]{geometry}
\usepackage[T1]{fontenc}
\usepackage{lmodern}
\usepackage{microtype}
\usepackage{xcolor}
\usepackage{algorithm}
\usepackage[noend]{algpseudocode}
\usepackage{ying}
\usepackage{amsmath,amssymb,amsthm,bm}
\usepackage{subcaption}
\usepackage{enumitem}
\usepackage{authblk}
\usepackage[colorlinks=true,linkcolor=blue,citecolor=blue,urlcolor=blue]{hyperref}
\usepackage[nameinlink]{cleveref}

\newtheorem{theorem}{Theorem}

\newtheorem{lemma}{Lemma}

\newtheorem{definition}{Definition}
\newtheorem{assumption}{Assumption}

\long\def\bb#1{#1}

\setcitestyle{round,semicolon}

\title{Conformal Policy Selection: When to Deploy AI Decisions for Targeted Intervention}
\author[a]{Junyu Cao}
\author[b]{Ruijiang Gao}
\author[a]{Maytal Saar-Tsechansky}
\author[c]{Mingzhang Yin}
\affil[a]{McCombs School of Business, The University of Texas at Austin, Austin, TX 78712}
\affil[b]{Naveen Jindal School of Management, The University of Texas at Dallas, Richardson, TX 75080}
\affil[c]{Warrington College of Business, University of Florida, Gainesville, FL 32611}
\makeatletter
\g@addto@macro\AB@affillist{\\[\affilsep]{\small Authors are listed in alphabetical order.}}
\makeatother
\date{}

\begin{document}
\maketitle

\begin{abstract}
\bb{Artificial-intelligence (AI)} policies are increasingly used to make consequential decisions in medicine, finance, education, and public services, yet deploying them responsibly requires knowing not only whether a policy helps \emph{on average}, but \emph{for whom}. Existing tools for off-policy evaluation and safe policy learning provide population-level guarantees, \bb{but some individuals may experience adverse treatment effects}. We recast targeted deployment as a multiple-testing problem: for each candidate unit, test whether the outcome under the AI policy would exceed a user-specified threshold, and control the false discovery rate (FDR) among the selected set. Our method, \emph{Conformal Policy Selection}, combines weighted conformal $p$-values that account for the distributional shift from behavior to target policy with a two-step selection-and-pruning procedure. We prove finite-sample FDR control and establish a \bb{doubly robust} asymptotic guarantee, where FDR remains controlled if either the outcome model or the propensity model is correctly specified. Synthetic experiments verify these guarantees across four misspecification regimes, and a semi-synthetic study built on real electronic-health-record covariates illustrates how the same FDR control can act as a deployment safety gate, bounding the harmful treatments that a globally trained policy would otherwise administer to a high-risk subpopulation where it misgeneralizes. The framework offers \bb{a principled method for safe AI deployment} in high-stakes applications.
\end{abstract}

\noindent\textbf{Keywords:} FDR Control; Policy Evaluation; AI Decisions; Safe Policy Learning

\section{Introduction}

Artificial-intelligence (AI) systems now assist or replace human judgment in consequential domains including clinical decision support, targeted advertising, credit underwriting, educational interventions, and resource allocation \citep{imbens2024causal, athey2021policy, kleinberg2018human}. Responsible deployment, however, requires addressing two distinct challenges. The first is validating whether a new policy improves outcomes on average. This objective is supported by a mature literature in off-policy evaluation \citep{athey2021policy, dudik2014doubly, chernozhukov2018double} and safe policy learning, which certifies population-level improvements with high probability \citep{thomas2015high, laroche2019safe, thomas2019preventing, benmichael2025safe}. Yet, an equally critical but conceptually distinct challenge is less explored: determining precisely \emph{for whom} the policy should be deployed\bb{.}

Population-level guarantees can hide substantial individual-level harm. A policy that raises average outcomes may nevertheless leave identifiable subgroups worse off \citep{kallus2019classifying, lei2023policy}, and the stakes of such within-policy heterogeneity are precisely what drives the move toward personalization: a treatment rule is valuable because it adapts to the individual. Deploying a universal rule is often neither practical nor ethical when the intervention carries risk. The decision-maker instead wishes to identify a \emph{subpopulation} of units for whom the new policy is confidently better than the status quo, and to deploy it selectively. This shifts the statistical problem from \emph{estimation} of a single number (the policy value $\mu(\pi)$) to a \emph{selective-inference} problem, which requires  \bb{simultaneously testing whether the policy's outcome} is greater than a deployment-relevant threshold, while controlling the rate of incorrect selections \citep{jin2023model, jin2023selection}.

Thresholding a predictive model to identify ``likely winners'' is the natural first instinct, but this approach does not provide a rigorous theoretical guarantee. \bb{Rankings from predictive models do not guarantee valid subgroup selection}, and standard \bb{machine-learning (ML)} calibration fails under the distributional shift inherent in off-policy settings, where the historical data were generated under a behavior policy $\pi_0$ and the selection must generalize to deployments of a different policy $\pi$ \citep{tibshirani2019conformal}. Individual-level conformal prediction intervals \citep{vovk2005algorithmic, angelopoulos2023conformal, lei2021conformal} yield marginally valid coverage. However, these intervals lack conditional coverage guarantees, and naive thresholding of such intervals does not control the \emph{false discovery rate} (FDR) of the resulting selection \citep{benjamini1995controlling}. Neither safe policy improvement nor \bb{decision rules based on conditional average treatment effects (CATEs)} address the multiple-testing structure directly, nor do they provide the false discovery guarantees under the joint uncertainty of outcome and propensity models.

\paragraph{Our contribution.}
We formalize targeted policy deployment as a multiple-testing problem and introduce \emph{Conformal Policy Selection}, a procedure that returns a set $\mathcal{S}\subseteq\{1,\ldots,m\}$ of test units predicted to clear a user-specified threshold under a future treatment policy (e.g., an AI policy) while guaranteeing that the expected fraction of false discoveries in $\mathcal{S}$ is at most a user-chosen level $q$. Our starting point is weighted conformal inference under covariate shift \citep{tibshirani2019conformal, jin2023model}, which we extend in two directions. First, we construct weighted conformal $p$-values tailored to the policy-evaluation setting, in which the shift from behavior to target policy has an explicit likelihood-ratio form that can be estimated from observational data. Second, we combine these $p$-values with a two-step selection-and-pruning procedure that yields \emph{finite-sample} FDR control in the calibrated-oracle setting. We further prove a \emph{\bb{doubly robust} asymptotic} guarantee that FDR control is preserved whenever \emph{either} the nonconformity score (outcome model) or the weight function (propensity model) is correctly specified. %

\section{Related Work}

\subsection{Safe Policy Learning}

Deploying a policy learned from observational data carries inherent risks. The policy may underperform the existing baseline due to model misspecification, distributional shift, or unobserved confounding. This has motivated a growing literature on safe policy learning, which seeks formal guarantees that a new policy will not cause harm before it is deployed \citep{garcia2015comprehensive, gu2024review}.

One line of work focuses on safe policy improvement, which ensures that the learned policy performs at least as well as the behavior policy with high probability. \bb{\citet{thomas2015high}} provide high-confidence lower bounds on off-policy value, enabling a practitioner to verify with \bb{probably approximately correct (PAC)} guarantees that a candidate policy meets a minimum performance threshold before deployment. \bb{\citet{laroche2019safe}} extend this idea by bootstrapping the baseline policy's value to certify improvement. The Seldonian framework of \bb{\citet{thomas2019preventing}} generalizes safe deployment to user-specified behavioral constraints, guaranteeing that undesirable outcomes occur with at most a prescribed probability. \bb{\citet{benmichael2025safe}} study safe policy learning under deterministic behavior policies, introducing a statistical safety property that limits the probability of deploying a policy worse than the existing one. Related approaches include empirical welfare maximization \citep{manski2004statistical, kitagawa2018should, athey2021policy} and doubly robust policy evaluation \citep{dudik2014doubly, chernozhukov2018double}, which improve the reliability of policy value estimates but do not directly enforce safety constraints.

In sequential and personalized settings, safe policy learning requires maintaining safety constraints throughout both the learning and deployment phases. \bb{\citet{cao2024safe} develop} a contextual constrained Markov decision process framework for personalized medical decision-making that achieves sub-linear regret while never violating safety requirements, providing the first theoretical result unifying contextual and safe reinforcement learning. \bb{\citet{lei2023policy}} address a related challenge of policy learning under biased sample selection, where the observed data may not represent the target population.

A common limitation of existing safe policy learning methods is that their safety guarantees are \emph{population-level} (i.e., bounding the average policy value, regret, or probability of overall policy improvement), and do not address which specific individuals should receive the new policy. In high-stakes domains, a policy that improves average outcomes may still harm specific subpopulations \citep{kallus2019classifying}. Our work provides a fundamentally different notion of safety. Rather than certifying the average performance of a policy, we control the FDR among individuals selected for AI-driven intervention, ensuring that the fraction of individuals who are harmed by incorrect policy assignment is bounded with finite-sample guarantees.

\subsection{Conformal Inference and Selection}

Conformal inference produces predictive intervals with distribution-free marginal coverage guarantees \citep{vovk2005algorithmic}. Recent work extends conformal inference from exchangeable data to settings with covariate shift \citep{tibshirani2019conformal} and applies it to off-policy prediction in contextual bandits \citep{Taufiq2022b}. 
\bb{Marginal coverage, however, does not by itself guarantee coverage within a selected subgroup.}

The use of conformal $p$-values together with the Benjamini--Hochberg (BH) procedure for finite-sample FDR control originates with outlier testing \citep{bates2023testing}. Conformal selection \citep{jin2023selection, Fisch2022w} builds on this idea to feed conformal $p$-values into the BH procedure for FDR-controlled subset selection. It has been applied to drug compound screening \citep{bai2025conformal} and certification of foundation model outputs \citep{Gui2024w}, though these applications usually require i.i.d. data. Our work builds on weighted conformalized selection \citep{jin2023model}, which extends conformal selection to covariate shift settings via auxiliary $p$-values. We adapt this framework to address the selective-inference problem in policy learning, generalizing the distributional shift to treatment-covariate pairs through the policy ratio. We further establish a new doubly robust asymptotic guarantee for FDR control over the individuals selected under the target policy. Concurrent work by \citet{jin2026conformal} uses conformal $p$-values to learn binary treatment rules that control the probability of individual harm relative to the control treatment for binary outcomes, whereas we take a given AI policy as fixed and control the FDR among the units selected for its deployment, testing whether each unit's realized outcome under the target policy exceeds a threshold.

\section{\bb{Preliminaries: Policy Evaluation}}
\subsection{Problem Setup}

Suppose we have access to a set of observed  tuples $\{X_i, T_i, Y_i\}_{i=1}^n$, 
where the covariates $X_i\in\mathcal{X}$, the treatment arm $T_i \in \cT:=\{0,\cdots, k-1\}$, and a scalar outcome $Y_i\in\mathbb{R}$. Using the potential outcome framework, we assume $Y_i = Y_i(T_i)$, i.e.,
\bb{the stable unit treatment value assumption (SUTVA)} \citep{rubin1980randomization}. 
We consider $Y$ as the reward, and higher reward is better. 
\bb{The observed tuples are generated} by a behavior policy $\pi_0:\cX\rightarrow \Delta^{|\cT|}$. \bb{The behavior policy may represent existing human decisions or an algorithmic policy.} 

The test samples $\{(X_{n+j},T_{n+j})\}_{j=1}^m$ have treatments assigned by a
\emph{target policy} $\pi:\mathcal{X}\to \Delta^{|\cT|}$, so
$T_{n+j} \sim \text{Categorical}(\pi(X_{n+j}))$. The outcome
$Y_{n+j}=Y_{n+j}(T_{n+j})$ is the potential outcome under treatment $T_{n+j}$,
which is \emph{unobserved}. Let $P$ denote the joint distribution of $(X,T,Y)$ when $X\sim P_X$,
$T\mid X\sim\pi_0(\cdot\mid X)$, and $Y=Y(T)$. 
The observed (calibration)
tuples are drawn i.i.d.\ from $P$. Let $Q$ denote the corresponding distribution when $X\sim Q_X$, $T\mid X\sim\pi(\cdot\mid X)$, and $Y=Y(T)$.
The test tuples are drawn i.i.d.\ from $Q$. 

\subsection{Policy Evaluation and Deployment}
A common practice in policy evaluation is to estimate the expected reward of a future policy $\pi$ using data collected under a different (behavior) policy $\pi_0$ to decide whether $\pi$ should be deployed in the future \citep{athey2021policy,imbens2024causal}. 

The estimand of interest is the average potential outcome under $\pi$, defined as
\[
\mu(\pi) := \EE_{X \sim Q_X}\left[ \EE\left[ Y(X,\pi(X)) \mid X \right] \right],
\]
where $Q_X$ denotes the marginal distribution of covariates in the target population, and $Y(X,T)$ is the potential outcome for the individual with feature $X$ under treatment $T$. Since we do not observe $Y(X, \pi(X))$ for all units, various estimators have been developed to address this missing data problem. Common approaches include direct method, which directly models $\EE[Y(X,T)|X,T]$ using the observational data, inverse propensity weighting (IPW) \citep{athey2021policy}, which reweights observed outcomes by the likelihood ratio between the target and behavior policies, and doubly robust (DR) estimators \citep{chernozhukov2018double}, which combine direct method with propensity weighting to improve efficiency and robustness. These estimators aim to provide unbiased or consistent estimates of $\mu(\pi)$ under suitable assumptions, such as correct specification of the propensity model or the outcome regression.

A key limitation of universal policy deployment is that, even if a new AI policy $\pi$ offers a marginal improvement in average reward over the baseline, it may still harm certain subpopulations. In other words, the aggregate gain can hide heterogeneous effects, where some individuals or groups experience worse outcomes under the new policy. This is particularly concerning in \textit{high-stakes applications}, where negative impacts on vulnerable or minority groups must be avoided. Therefore, rather than deploying $\pi$ for all, it is desirable to identify and serve only those individuals for whom the new policy is likely to result in a benefit, ensuring that interventions are both effective and equitable.

\section{Policy Selection with FDR Control}

Given a future policy $\pi: \mathcal{X}\rightarrow \Delta^{|\cT|}$, 
the goal is to identify a subset $\cS \subseteq \cD_\test$, where $\cD_{test}=\{n+1,\cdots, n+m\}$,  
with outcomes under the future policy $\pi$ above some known, potentially random, 
thresholds $\{c_{ j}\}_{j\in \cD_\test}$ with guaranteed probability. To draw statistical evidence for large outcomes in $\cD_\test$,  
we assume we hold an independent set of 
calibration data $\cD_\calib = \{(X_i,T_i, Y_i)\}_{i=1:n}$, whose treatments are also generated by the behavior policy $\pi_0$. 

We are interested in finding as many units $j\in \cD_\test$ as possible 
with $Y_{j}(T_j)>c_j$, while controlling the FDR 
below a \bb{prespecified} level $q\in (0,1)$. 
The FDR is the expected fraction  of false discoveries, where a false discovery is a selected unit for which $Y_j \le c_j$, i.e., the outcome turns out to fall below the threshold. 
Formally, we wish to find $\cS$ obeying 

\#\label{eq:fdr}
\fdr := \EE\Bigg[  \frac{\sum_{j\in\cD_\test} \ind\{j\in \cS,Y_{j}\leq c_{ j}\}}{\max\{1,|\cS|\}}  \Bigg] \leq q,
\#
where $q$ is the nominal target FDR level.

\subsection{Weighted Conformal Selection}

First, we introduce the weighted conformal selection.
\bb{\citet{jin2023model} propose} to choose the units $j\in \cD_{test}$ through testing the random hypotheses
\[
H_j: Y_{n+j}\leq c_j, \ j\in \cD_{test}.
\]
\subsubsection{Conformal $p$-values}
\bb{To test these hypotheses, we use the \emph{conformal $p$-values} \citep{jin2023selection} to quantify evidence against the null hypotheses.}
Suppose we are given a prediction model obtained from a training procedure that is independent of both the calibration and test samples. 
Conditioning on the training process, we treat this prediction model as fixed. Based on this model, we define a nonconformity score 
\[
V: \mathcal{X} \times \cT\times \mathcal{Y} \to \mathbb{R},
\]
which quantifies how well a response value $y$ conforms to the model’s prediction at input $x$ with treatment $t$. For instance, if the model provides a point prediction $\hat{\mu}: \mathcal{X}\times \cT \to \mathbb{R}$, one natural choice is 
$V(x,t, y) = y - \hat{\mu}(x,t)$.

If a test response $Y_{n+j}$ were observed, we could compute the nonconformity scores 
\[
V_i = V(X_i, T_i, Y_i), \quad i=1,\dots,n, 
\]
\[
V_{n+j} = V(X_{n+j},T_{n+j}, Y_{n+j}), j=1,\cdots, m.
\]
The associated conformal $p$-value is then defined as
\[
p_j^\ast = \frac{\sum_{i=1}^n \ind\{V_i < V_{n+j}\} + U_j \cdot (1+\sum_{i=1}^n \ind\{V_i = V_{n+j}\}) }{n+1},
\]
where $U_j \sim \mathrm{Unif}[0,1]$ are i.i.d.\ random variables introduced to break ties. When the test sample $(X_{n+j}, T_{n+j}, Y_{n+j})$ is exchangeable with the calibration data, it holds that $p_j^\ast \sim \mathrm{Unif}[0,1]$.

\bb{Validity of the conformal $p$-values relies on the following monotonicity property.}

\begin{definition}
A score function $V(\cdot,\cdot,\cdot): \cX\times \cT\times \cY\rightarrow \mathbb{R}$ is monotone if $V(x,t,y)\leq V(x,t,y')$ holds for any $x\in \cX$, $t\in \cT$, and any $y,y'\in \cY$ satisfying $y\leq y'$.
\end{definition}

However, in policy learning problems, the test samples receive treatments from a future policy $\pi$, which is different from the behavior policy $\pi_0$. 
We can formulate the policy difference as the density ratio between the target and behavior distributions. Specifically, the weight function is defined as

\begin{equation}\label{eq:shift-xt}
  \frac{\mathrm{d}Q}{\mathrm{d}P}(x,t,y)
  \;=\; w(x,t) \;=\; \frac{\mathrm{d}Q_X}{\mathrm{d}P_X}(x)\,
  \frac{\pi(t\mid x)}{\pi_0(t\mid x)},
\end{equation}

When the
covariate marginals agree ($Q_X=P_X$), $w(x,t)=\pi(t|x)/\pi_0(t|x)$
reduces to the policy ratio. To simplify the exposition, we assume $Q_X=P_X$ for the remainder of the paper. The general case is recovered by multiplying the policy ratio by the covariate density ratio $Q_X(x)/P_X(x)$.

This adjustment ensures that the conformal $p$-values remain valid even when the covariate distribution shifts between the calibration and test samples. 
Compute $V_i=V(X_i,T_i,Y_i)$ for $i=1,\dots,n$
and $\hat{V}_{n+j} = V(X_{n+j},T_{n+j}, c_{n+j})$ for $j=1,\dots,m$. The weighted conformal \bb{$p$-values} are 
\begin{align}
\label{eq:weighted_pval}
{p}_j = \frac{\sum_{i=1}^n w(X_i, T_i)\ind{\{V_i < \hat{V}_{n+j} \}}+  w(X_{n+j}, T_{n+j}) }{\sum_{i=1}^n w(X_i, T_i) + w(X_{n+j}, T_{n+j})}.
\end{align}
For each $j$, we compute a set of $m$ auxiliary \bb{$p$-values}
\begin{align}
 &{p}_\ell^{(j)} := \frac{
    \substack{
        \sum_{i=1}^n w(X_i, T_i)\,\mathbf{1}\{V_i < \widehat{V}_{n+\ell}\} \\
        + w(X_{n+j}, T_{n+j})\,\mathbf{1}\{\widehat{V}_{n+j} < \widehat{V}_{n+\ell}\}
    }
}{
    \sum_{i=1}^n w(X_i, T_i) + w(X_{n+j}, T_{n+j})
}, \quad \ell \neq j \nonumber \\
 &p_\ell^{(j)} := 0, \quad \ell = j.
 \label{eq:aux}
\end{align}

Next, we discuss how to leverage the constructed \bb{$p$-values} for the subgroup selection with rigorous \bb{FDR} control.  

\subsection{Conformal Policy Selection}
\bb{In this section, we introduce the Conformal Policy Selection algorithm.}
In our setting, the target policy $\pi$ differs from the behavior policy $\pi_0$. At the same time, the test response under a different treatment is unobservable. 

\subsubsection{Construct the Selection Set}

With the weight function $w(x,t)$, the score $V_i = V(X_i, T_i, Y_i)$ for each calibration data point $i$, and the approximate score $\hat{V}_{n+j} = V(X_{n+j},T_{n+j}, c_{n+j})$ for the test data $j$, we compute the conformal $p$-values $\{p_j\}_{j=1}^m$ with \Cref{eq:weighted_pval}, \bb{which satisfy null-event super-uniformity:} $\mathbb{P}(p_j \leq t, \, Y_{n+j} \leq c_{n+j}) \leq t$ for all $t \in [0,1]$. \bb{Small values of $p_j$ provide evidence against $H_j$ and support selecting the $j$-th test unit.} 

We use the auxiliary \bb{$p$-values} $\{p_\ell^{(j)}\}_{j,\ell=1}^m$ from \Cref{eq:aux} to construct the thresholds, which measure whether $p_j$ is sufficiently small for selecting the $j$-th test point. Specifically, we conduct the \bb{Benjamini--Hochberg} (BH) procedure over the set $\{p_\ell^{(j)}\}_{\ell=1}^m$ with nominal level $q$ and denote the size of the selected set as $K_j$. However, simply selecting unit $j$ with $p_j \leq q K_j / m$ may not guarantee FDR control. This is because the unit is selected not only because its \bb{$p$-value} is small, but also because $K_j$ accidentally is too large due to the correlation across $K_j$s. To address this, we sample $u \sim \text{Uniform}(0,1)$, and take the selection set as
\begin{align}
    \cS = \{j: p_j \leq q K_j / m \;\text{and}\; u \cdot K_j \leq r^* \}
    \label{eq:set}
\end{align}
where $r^* = \max \{r: \sum_{j=1}^m \mathbb{I}\{p_j \leq q K_j / m, ~~ u \cdot K_j \leq r \} \geq r \}$ is the maximal selection set size that the data can support.

\begin{algorithm}[h]
  \caption{Conformal Policy Selection}\label{alg:bh}
  \begin{algorithmic}[1]
    \Require {\raggedright \bb{Calibration data $\{(X_i, T_i, Y_i)\}_{i=1}^n$,
             test data $\{(X_{n+j},T_{n+j})\}_{j=1}^m$,
             thresholds $\{c_{n+j}\}_{j=1}^m$,
             weight function $w(\cdot)$,
             FDR target $q\in(0,1)$, monotone nonconformity score $V\colon \cX\times\cT\times\cY\to \RR$\par}}
    \State Compute $V_i = V(X_i,T_i,Y_i)$ for $i=1,\dots,n$,
           and $\hat{V}_{n+j}= V(X_{n+j},T_{n+j}, c_{n+j})$ for $j=1,\dots,m$.
    \State Construct weighted conformal \bb{$p$-values} $\{ p_j\}_{j=1}^m$ as in \Cref{eq:weighted_pval}.
    \For{$j=1,\dots,m$}
        \State Compute auxiliary \bb{$p$-values} $\{p_\ell^{(j)}\}_{\ell=1}^m$ as in \Cref{eq:aux}.
        \State Sort $\{p_\ell^{(j)}\}_{\ell=1}^m$ increasingly as $\{p_{(k)}^{(j)}\}_{k=1}^m$.
        \State Compute $K_j = \max\bigl\{k : p_{(k)}^{(j)} \le qk/m\bigr\}$.
    \EndFor
    \State Sample $u \sim \text{Uniform}(0,1)$
    \State Compute $r^*$ and the set $\cS$ as in \Cref{eq:set}.
    \Ensure Selection set $\cS$.
  \end{algorithmic}
\end{algorithm}

Algorithm \ref{alg:bh} uses the weights to compute the auxiliary \bb{$p$-values} in \bb{\Cref{eq:aux}}. In practice, these weights may or may not be known to the decision-maker. In what follows, we discuss the guarantees in two cases: when the weights are known, and when they are misspecified.

\subsection{Non-asymptotic FDR Guarantee When Weight Is Known}
Denote $Z_i = (X_i, T_i, Y_i)$ for $1 \leq i \leq n+m$, and $\tilde{Z}_{n+j} = (X_{n+j}, T_{n+j}, c_j)$ for $1 \leq j \leq m$.

\begin{theorem}\label{thm:calib_ite}
Suppose:
  \begin{enumerate}[label=(\roman*),noitemsep]
    \item $\{Z_i\}_{i=1}^n\overset{\mathrm{iid}}{\sim}P$ and
          $\{Z_{n+j}\}_{j=1}^m\overset{\mathrm{iid}}{\sim}Q$, with
          joint covariate-treatment shift~\Cref{eq:shift-xt}.
    \item $V:\mathcal{X}\times\mathcal{T}\times\mathcal{Y}\to\mathbb{R}$
          is monotone (nondecreasing in its third argument).
    \item For each $j=1,\ldots,m$, the samples in
          $\{Z_1,\ldots,Z_n,Z_{n+j}\}\cup\{\widetilde{Z}_{n+\ell}\}_{\ell\neq j}$
          are mutually independent.
  \end{enumerate}
  Then for any $q\in(0,1)$, the output $\mathcal{S}$ of Algorithm~1
  satisfies $\mathrm{FDR}\leq q$.
\end{theorem}

\medskip
\noindent\textit{Discussion.}
Theorem~\ref{thm:calib_ite} provides a \emph{finite-sample}, non-asymptotic guarantee: $\mathrm{FDR}\leq q$ holds exactly for every $n$ and $m$, with no requirement on sample size or distributional regularity beyond the stated conditions. This is a direct consequence of the weighted exchangeability structure (Appendix~\ref{app:proofs}), which ensures that the oracle conformal $p$-values are conditionally super-uniform even under the joint covariate-treatment shift~\Cref{eq:shift-xt}.

Several aspects of the assumptions merit attention.
\emph{Condition~(i)} permits the calibration distribution $P$ and the target distribution $Q$ to differ in both the covariate marginal and the treatment-assignment mechanism, provided the density ratio $w(x,t)=\mathrm{d}Q/\mathrm{d}P(x,t)$ is known. When the covariate marginals coincide ($Q_X=P_X$), $w$ reduces to the policy ratio $\pi(t|x)/\pi_0(t|x)$, recovering the standard off-policy evaluation setup.
\emph{Condition~(ii)} (monotonicity of the score function) is mild: it is satisfied by any residual-type score $V(x,t,y)=y-\hat\mu(x,t)$ or any quantile-based score whose third argument is the response.
\emph{Condition~(iii)} (mutual independence) is ensured by design whenever the calibration data are collected independently of the test features and thresholds (e.g., in the common setting where the test sample is drawn after calibration is complete and the thresholds $c_j$ are prespecified).

The proof proceeds in two steps. First, the pruning rule in Algorithm~\ref{alg:bh} reduces the overall FDR to a sum of per-unit terms. Second, for each unit under the null, we bound the observed conformal $p$-value by an oracle version based on the unit's true, unobserved outcome. Because the corresponding selection threshold is determined by the calibration sample and the observed information from the other test units, it is independent of that outcome. Weighted exchangeability then implies that the oracle $p$-value is super-uniform, giving a $q/m$ bound for each unit. Summing these bounds yields $\mathrm{FDR}\leq q$.

\subsection{Asymptotic Doubly Robust Guarantee When the Weight Model Is Misspecified}

Theorem~\ref{thm:calib_ite} delivers exact, finite-sample FDR control when the density ratio $w(x,t)=\mathrm{d}Q/\mathrm{d}P(x,t)$ is known. In practice, this ratio must usually be estimated and may be misspecified (e.g., humans' decision probabilities are not known). The key question is whether such estimation bias will harm validity. The next result shows that it does not when the nonconformity score is correctly specified. In particular, if the score is the true residual $V(x,t,y)=y-f(x,t)$ under a model $Y=f(X,T)+\epsilon$ with $\epsilon$ independent of $(X,T)$, then the null score distribution is invariant to $(X,T)$. In that regime, any bounded working weight function merely reweights the calibration sample without changing the limiting null distribution of the score. In this sense, the procedure is doubly robust: a correct outcome model can absorb misspecification in the weight model.

\begin{theorem}\label{thm: asymptotic dr}
Suppose $\{Z_i\}_{i=1}^n \overset{\mathrm{iid}}{\sim} P$ and
$\{Z_{n+j}\}_{j=1}^m \overset{\mathrm{iid}}{\sim} Q$ satisfy the joint
covariate-treatment shift~\Cref{eq:shift-xt}, and condition~(iii) of
Theorem~\ref{thm:calib_ite} holds. Assume moreover that the \bb{data-generating process (DGP)} satisfies
$Y=f(X,T)+\epsilon$ with $\epsilon$ independent of $(X,T)$, and that the
score is the oracle residual $V(x,t,y)=y-f(x,t)$. Let Algorithm~\ref{alg:bh}
use any working weight function $w^\dagger:\mathcal{X}\times\mathcal{T}\to(0,\infty)$
such that
\[
  \bar{w}^{-1}\;\leq\;\frac{w^\dagger(X_i,T_i)}{w^\dagger(X_{n+j},T_{n+j})}\;\leq\; \bar{w}
\]
almost surely for all $i=1,\ldots,n$ and $j=1,\ldots,m$ (equivalently, the working weights lie in a compact interval $[1/C,C]$). Then, for any
$q\in(0,1)$, the output $\cS$ of Algorithm~\ref{alg:bh} satisfies
\[
  \limsup_{n\to\infty}\mathrm{FDR}\leq q.
\]
\end{theorem}

\medskip
\noindent\textit{Discussion.}
Theorem~\ref{thm: asymptotic dr} no longer requires weights to
\bb{equal the true density ratio}. The only restriction is the bounded-ratio
condition, which rules out a single observation dominating the weighted
\bb{empirical cumulative distribution function (CDF)}. Under the residual-score model $V=\epsilon$, the score is
independent of $(X,T)$, so reweighting by any bounded $w^\dagger$ leaves the
limiting null score distribution unchanged.

This is the doubly robust feature of the procedure: if the outcome model is
correct, asymptotic FDR control survives even under a misspecified weight
model. In practice, the bounded-ratio condition can be enforced by clipping
estimated weights to a compact interval such as $[1/C,C]$ for some constant
$C>1$.

The price of this robustness is that the guarantee is asymptotic, rather than
finite-sample as in Theorem~\ref{thm:calib_ite}. Nevertheless, it covers a
practically important regime in which the score model is reliable while the
weight model cannot be estimated very accurately. 

\subsection{Asymptotic Guarantee: When Weights and Estimators Are Asymptotically Correct}

Theorems~\ref{thm:calib_ite} and~\ref{thm: asymptotic dr} each rely on one correctly specified nuisance function. Theorem~\ref{thm:calib_ite} assumes the density ratio $w$ is
known exactly, while Theorem~\ref{thm: asymptotic dr} allows an arbitrary
 weight function but requires the oracle residual score $V(x,t,y)=y-f(x,t)$. In
practice, neither the propensity score $e(x)$ nor the regression function
$f(x,t)$ is available, and both must be estimated from data. We therefore
consider the plug-in version of Algorithm~\ref{alg:bh}, in which the weights
and the test-side scores are built from estimators fitted on an independent
training sample, and ask whether FDR control is recovered as these estimators
improve as the sample size grows. The following assumption formalizes the required consistency.

\begin{assumption}\label{assump: weight and estimator}
  \bb{The DGP} satisfies $Y = f(X,T) + \epsilon$ with $\epsilon$ independent of $(X,T)$, and the
  nonconformity score is $V(x,t,y)=y-f(x,t)$.
  The algorithm uses weights $\hat{w}_N(x,t)=1/\hat{e}_N(x)$ and score
  $\hat{V}_{n+j}^N = c_{n+j}-\hat{f}_N(X_{n+j},T_{n+j})$ based on
  estimators $\hat{e}_N, \hat{f}_N$ trained on an independent dataset of sample size $N$.
  As $N\to\infty$, uniformly over the support of $(X,T)$:
  \[
    \left|\frac{1}{\hat{e}_N(X)}-\frac{1}{e(X)}\right|\to 0,\qquad
    \bigl|\hat{f}_N(X,T)-f(X,T)\bigr|\to 0 \quad \text{a.s.},
  \]
  where $e(x):=\pi_0(T=1\mid X=x)$ is the true propensity score.
\end{assumption}

\begin{theorem}\label{thm:estimated}
  Suppose Assumption~\ref{assump: weight and estimator} and
  the conditions of Theorem~\ref{thm:calib_ite} hold.
  For any $q\in(0,1)$, the output $\cS$ of Algorithm~\ref{alg:bh} satisfies
  $\limsup_{N\to\infty}\mathrm{FDR}\leq q$.
\end{theorem}

\medskip
\noindent\textit{Discussion.}
Theorem~\ref{thm:estimated} addresses the most practically relevant setting, in which \emph{both} the propensity score $e(x)$ and the conditional mean function $f(x,t)$ are unknown and must be estimated from data. The key mechanism is \emph{sample splitting}: the nuisance estimators $\hat{e}_N$ and $\hat{f}_N$ are trained on an independent sample of size $N$, separate from the calibration and test data. This independence ensures that, conditional on the training sample, the estimated score $\hat{V}_{n+j}^N = c_{n+j} - \hat{f}_N(X_{n+j},T_{n+j})$ and the estimated weights $\hat{w}_N(x,t) = 1/\hat{e}_N(x)$ are \bb{fixed functions. The asymptotic guarantee additionally requires the consistency conditions in Assumption~\ref{assump: weight and estimator}.}

Assumption~\ref{assump: weight and estimator} requires \emph{uniform almost-sure consistency} of both nuisance estimators: $\hat{e}_N(x)\to e(x)$ and $\hat{f}_N(x,t)\to f(x,t)$ uniformly as $N\to\infty$. Notably, no specific convergence rate (e.g., $N^{-1/2}$) is required, making the result compatible with flexible nonparametric or machine-learning estimators such as random forests, neural networks, or gradient boosting, provided they are consistent.

The residual score $V(x,t,y) = y - f(x,t)$ plays a dual role. First, it is nondecreasing in $y$, satisfying the monotonicity condition~(ii) of Theorem~\ref{thm:calib_ite}. Second, under the DGP $Y = f(X,T) + \epsilon$ with $\epsilon$ independent of  $(X,T)$, the oracle residuals $\{V_i = \epsilon_i\}$ are i.i.d.\ and independent of covariates. As $\hat{f}_N\to f$, the estimated residuals inherit this favorable structure, \bb{supporting asymptotic FDR control under the stated assumptions.}

We emphasize that both nuisance functions must be consistently estimated. Unlike Theorem~\ref{thm: asymptotic dr}, which tolerates weight misspecification, the present result requires convergence of $\hat{f}_N$ to the true regression function to ensure that the nonconformity scores remain well-calibrated. 

\section{Numerical Experiments}

We empirically examine the proposed method with synthetic and real-world data. We compare two procedures: Weighted BH (WBH), which applies the Benjamini--Hochberg procedure directly to the weighted conformal $p$-values, and Weighted Conformal Selection (WCS), the pruned procedure of Algorithm~\ref{alg:bh}. The finite-sample guarantee of Theorem~\ref{thm:calib_ite} applies to WCS. WBH is a simpler and typically more powerful variant that we report as a practical baseline and whose FDR control we assess empirically. We report the empirical FDR and power across a range of nominal levels $q \in [0.05, 0.95]$, averaged over multiple replications with error bars indicating one standard deviation.
The replication code is available at \url{https://github.com/ruijiang81/conformal-policy-selection}. 

\subsection{Synthetic Data}

\subsubsection{Data Generating Process}
Features are drawn i.i.d.\ as $X\sim\mathcal{N}(0,I_p)$ with $p=5$. The behavior policy is a logistic model $\pi_0(T=1\mid X)=\sigma(3\,\beta_T^\top X)$ with $\beta_T \sim \mathcal{N}(0,I_p)$ fixed across replications and propensities clipped to $[0.1,0.9]$ for overlap. Potential outcomes follow linear models $Y(t) = \beta_t^\top X + \epsilon_t$, $t\in\{0,1\}$, with $\epsilon_t\sim\mathcal{N}(0,\sigma^2)$, $\sigma=1$, and $\beta_0,\beta_1\sim\mathcal{N}(0,I_p)$. The observed outcome is $Y = (1-T)Y(0) + T\,Y(1)$, yielding tuples $(X,T,Y)$ from the behavior policy.
The AI policy $\pi$ is trained on an independent sample of $n_{\mathrm{train}}=200$ observations via linear regression of $Y$ on $(X,T)$ (the concatenated design $[X\ T]$), producing an estimate $\hat\mu(x,t)$. The policy takes $\pi(x)=\ind\{\hat\mu(x,1)>\hat\mu(x,0)\}$. Calibration and test samples of size $n=500$ and $m=100$ are drawn independently from the same distribution. Nonconformity scores on the calibration set use the true outcome model, $V_i = Y_i - \beta_{T_i}^\top X_i$, and test-side scores are defined relative to the threshold, $\widehat V_{n+j} = c_{n+j} - \hat\mu\bigl(X_{n+j},\pi(X_{n+j})\bigr)$. Likelihood-ratio weights for shift from behavior to target use $w_i = \ind\{T_i = \pi(X_i)\}/\pi_0(T_i\mid X_i)$. We set the threshold $c=0$ (select units with positive outcome under $\pi$) and average results over $400$ independent replications.

\subsubsection{FDR Control with Correctly Specified Models}

Figure~\ref{fig:synthetic} shows the empirical FDR and power when both the nonconformity scores (using the true outcome model) and the weights (using the true propensity) are correctly specified. The FDR is well-controlled below the nominal level across all $q$, consistent with the finite-sample guarantee of Theorem~\ref{thm:calib_ite}. At the same time, the power increases steadily, reaching approximately $0.80$ at $q=0.95$, indicating effective identification of individuals who benefit from the AI policy.

\begin{figure}[!t]
    \centering
    \begin{subfigure}[b]{0.48\linewidth}
        \centering
        \includegraphics[width=\linewidth]{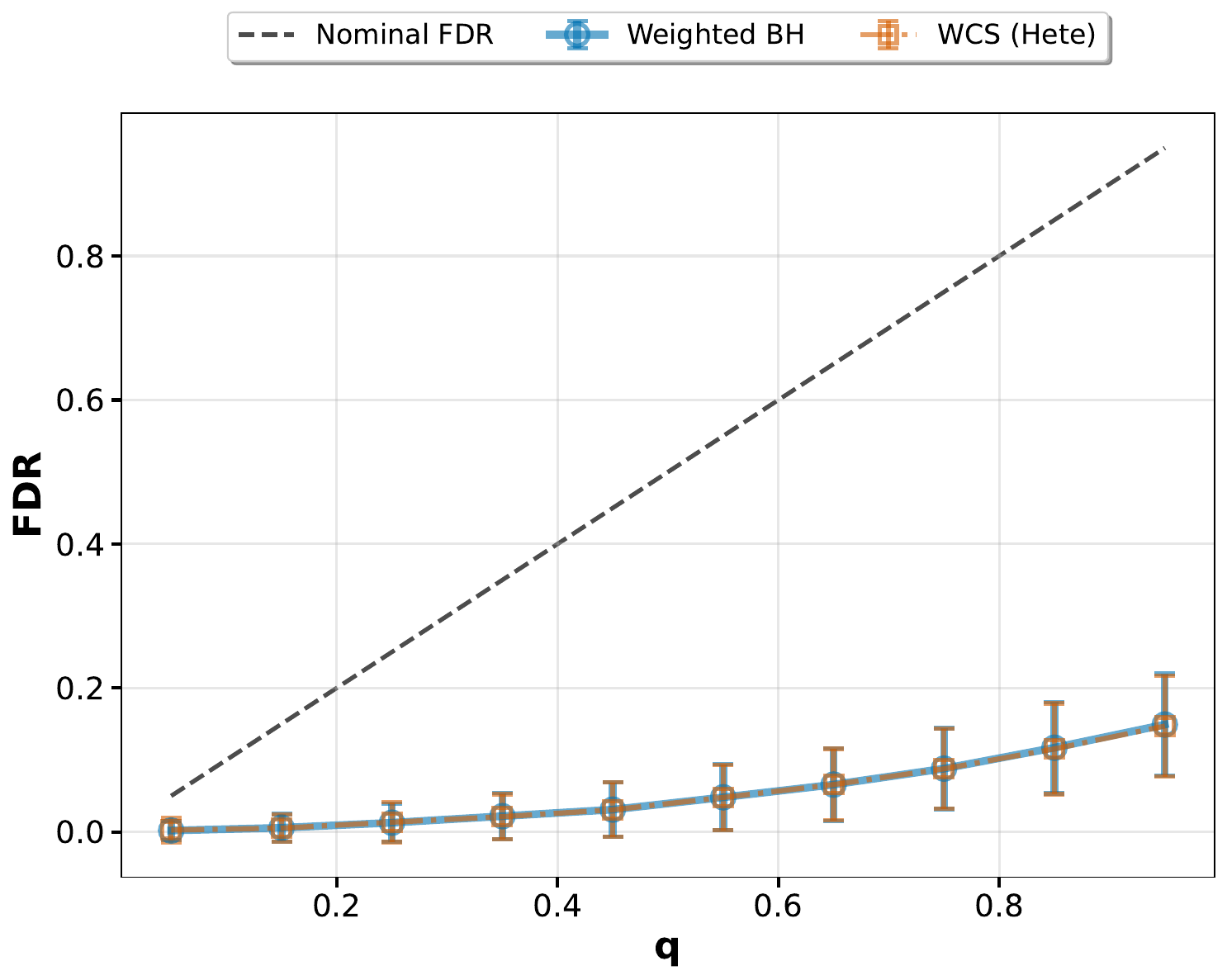}
    \end{subfigure}
    \hfill
    \begin{subfigure}[b]{0.48\linewidth}
        \centering
        \includegraphics[width=\linewidth]{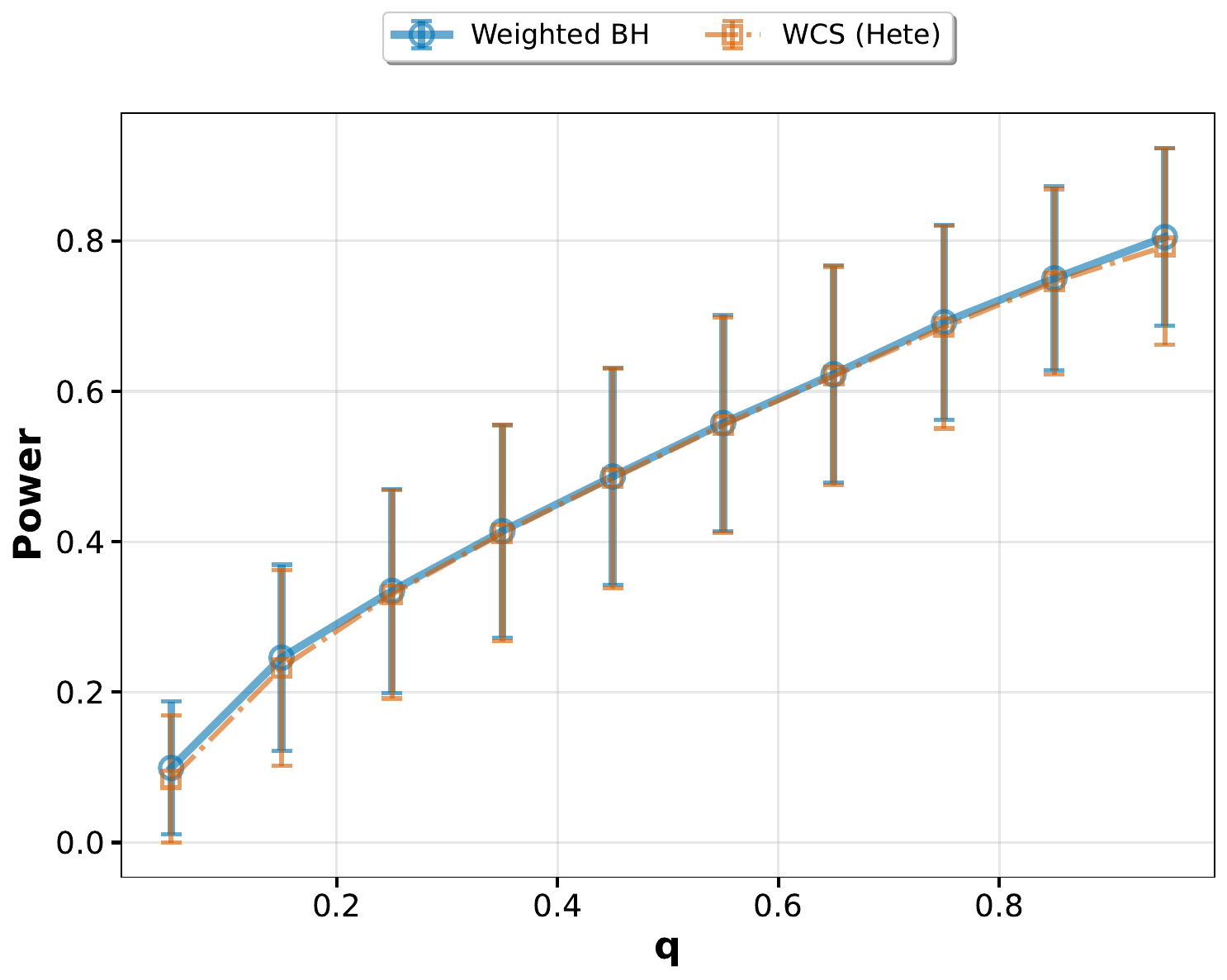}
    \end{subfigure}
    \caption{FDR and power with correctly specified scores and weights on synthetic data ($n=500$, $m=100$, 400 replications). The FDR remains well below the nominal level $q$ (dashed line) across all target levels, while power increases steadily.}
    \label{fig:synthetic}
\end{figure}

\subsubsection{Doubly Robust FDR Control}

A key property of our method is its robustness to potential model misspecification. We systematically evaluate four scenarios that share the same simulated data per replication, differing only in the score and weight functions fed to the selection procedure:
\begin{enumerate}
    \item[(1)] \textbf{Both correct}: scores use the true regression coefficients $\beta_0,\beta_1$, and weights use the true propensity $\pi_0(\cdot\mid X)$.
    \item[(2)] \textbf{Correct weights, wrong scores}: scores use misspecified coefficients $\beta_0',\beta_1'\sim\mathcal{N}(0,3^2 I_p)$ applied to the calibration and test sets, while the weights use the true propensity. Therefore, the outcome model is inaccurate.
    \item[(3)] \textbf{Wrong weights, correct scores}: scores use the true coefficients, while the weights use a misspecified propensity obtained by flipping and amplifying the logit ($\tilde\pi_0(T{=}1\mid X)=\sigma(-8\,\beta_W^\top X)$, clipped to $[0.02,0.98]$).
    \item[(4)] \textbf{Both wrong}: the misspecified scores of~(2) and the misspecified weights of~(3) together.
\end{enumerate}
Each of Scenarios (1)-(3) satisfies the hypotheses of one of our guarantees. When the weights are the true policy ratio and the residual score is applied consistently to calibration and test sets, Theorem~\ref{thm:calib_ite} applies and FDR is controlled regardless of whether the outcome model is accurate (Scenarios~1 and~2). When the residual score is the true one, Theorem~\ref{thm: asymptotic dr} applies and FDR is controlled even under misspecified weights (Scenarios~1 and~3). Only Scenario~(4), where the weights are wrong \emph{and} the score is misspecified, does not satisfy the FDR guarantee.

Figure~\ref{fig:dr_panel} presents the FDR across all four scenarios. FDR is controlled below the nominal level whenever \emph{either} the score or the weight model is correctly specified (Scenarios~1, 2, and~3). When both are misspecified (Scenario~4), the empirical FDR exceeds the nominal level, reaching about $0.30$ at $q=0.05$ and remaining elevated at small $q$.

Scenario~2 is the key \bb{doubly robust} case. The outcome model is misspecified, yet because the residual score is applied identically to calibration and test, the calibration and test scores remain weighted-exchangeable, and correct weights alone are enough for FDR control (Theorem~\ref{thm:calib_ite}). Misspecifying the outcome model reduces power but does not break validity. Control fails only when the weights are also wrong (Scenario~4), so that neither guarantee applies.

\begin{figure}[!t]
    \centering
    \includegraphics[width=\linewidth]{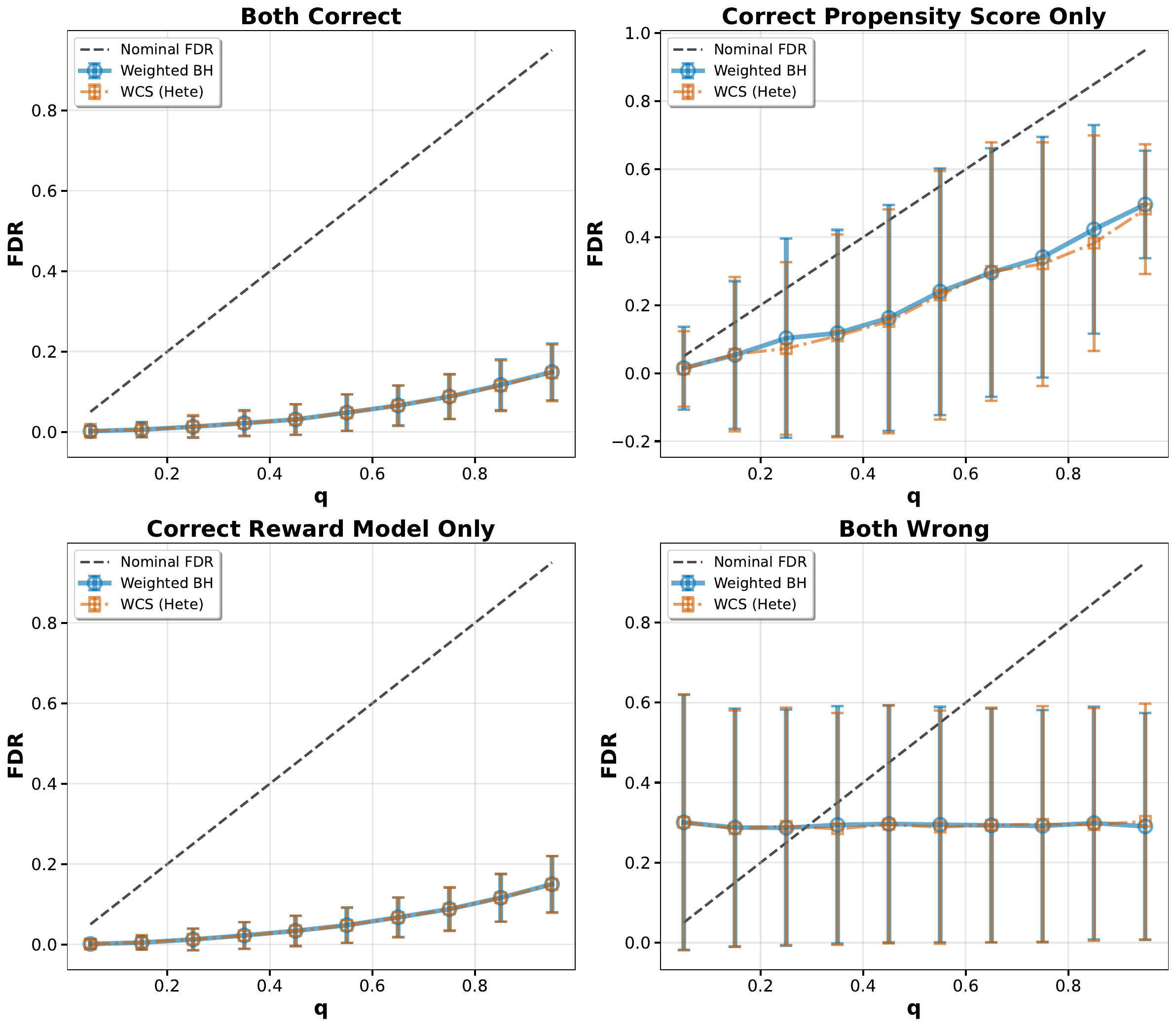}
    \caption{Doubly robust FDR control experiment. Each panel shows the empirical FDR across nominal levels $q$ under different combinations of score and weight specification. \textbf{Top left}: both correct, FDR controlled. \textbf{Top right}: correct weights, misspecified (but consistently applied) scores, FDR controlled. \textbf{Bottom left}: wrong weights, correct scores, FDR controlled. \textbf{Bottom right}: both wrong, FDR not controlled. The dashed line indicates the nominal FDR level $q$. Results averaged over 400 replications with $n=500$ calibration and $m=100$ test observations.}
    \label{fig:dr_panel}
\end{figure}

Table~\ref{tab:dr_summary} summarizes the mean FDR at representative nominal levels. Correct specification of \emph{either} the nonconformity score or the weight model is sufficient for FDR control.

\begin{table}[!t]
    \centering
    \caption{Mean empirical FDR at selected nominal levels $q$ across the four doubly robust scenarios (WBH method, 400 replications).}
    \begin{tabular}{lcccc}
        \toprule
        Scenario & $q=0.15$ & $q=0.25$ & $q=0.55$ & $q=0.95$ \\
        \midrule
        Both correct      & 0.005 & 0.013 & 0.048 & 0.149 \\
        Correct W, wrong S & 0.055 & 0.083 & 0.225 & 0.502 \\
        Wrong W, correct S & 0.005 & 0.013 & 0.048 & 0.150 \\
        Both wrong         & 0.287 & 0.300 & 0.295 & 0.297 \\
        \bottomrule
    \end{tabular}
    \label{tab:dr_summary}
\end{table}

\subsection{Real-World Data: Diabetes Prevention from Electronic Health Records}

We now apply the framework to a diabetes-prevention decision on real electronic health records, using it to illustrate a high-stakes application for FDR control.  We aim to use conformal policy selection as a \emph{deployment safety gate} that prevents the AI policy trained on the whole population from harming a subpopulation. A complementary semi-synthetic study on selecting who benefits from a job-training program is given in Appendix~\ref{sec:lalonde}.

\subsubsection{Dataset and setting}
We use the longitudinal \bb{electronic health record (EHR) dataset of \citet{kraus2024data}} from a national Israeli health insurer. There are $89{,}999$ prediabetic patients over \bb{2003--2013} ($540{,}419$ patient-years, $305$ features after encoding). The outcome $\tilde Y = 1 - Y$ indicates avoidance of type II diabetes onset (ICD-9 $250$) within a $3$-year window (``higher is better''), and the treatment $T=1$ denotes a preventive antidiabetic prescription (Glucophage or Glucomin). We study the clinically high-risk subgroup with age $\ge 60$ and BMI $\ge 30$, for which $n = 5{,}993$ patient-years and $8.3\%$ are treated.

In this dataset, antidiabetics are prescribed reactively, especially when a patient is already deteriorating, so $\Pr(\text{onset}\mid T{=}1)=97\%$ and $\Pr(\text{onset}\mid T{=}0)=31\%$. This confounding by indication is not removable by covariate adjustment. Fitting an S-learner, T-learner, X-learner, or causal forest yields a \emph{negative} estimated treatment effect for every patient, and even within propensity-matched strata the treated avoid onset far less often (Appendix~\ref{sec:si-confounding}). We treat the fitted values as the true potential outcomes. \bb{Therefore, a benefit-maximizing policy treats no one.} \bb{Because counterfactual outcomes are unobserved, we do not draw conclusions about the actual benefits of antidiabetic prescriptions; our conclusions concern only the semi-synthetic setup.} 
The standard clinical rule of treating the \bb{highest-risk decile is} a policy driven by a plain onset-probability model $\widehat{\Pr}(\text{onset}\mid x)$, and flags about $10\%$ of this subgroup for a treatment. Due to our semi-synthetic setup, this policy is harmful.  

\subsubsection{Conformal selection as a deployment gate}
We deploy the AI policy $\pi$ only to the units that the conformal policy selection method chooses. Every other unit keeps the status quo. FDR control at level \(q\) bounds the expected proportion of selected deployments whose outcomes fall at or below the prespecified threshold by \(q\), with this proportion defined as zero when no units are selected. Real counterfactuals are unobservable, so the true FDR cannot be measured on the raw data. We therefore evaluate against a semi-synthetic ground truth, in which a causal forest fit on the real covariates supplies the conditional treatment effect $\tau(x)$ and a random forest supplies the control surface $\mu_0(x)$, giving oracle potential-outcome means $\mu_0(x)$ and $\mu_1(x)=\mu_0(x)+\tau(x)$ from which we draw continuous outcomes $Y(t)=\mu_t(x)+\varepsilon$. This provides both counterfactuals for every unit, so the reported FDR and the count of deployed harmful treatments are exact. The nonconformity score uses a random-forest S-learner decoupled from the risk-based policy, and the weights are the estimated policy-shift likelihood ratios. Full details are in Appendix~\ref{sec:si-diabetes}.

\subsubsection{Results and implication}
Figure~\ref{fig:diabetes} reports the outcome over $20$ replications. The empirical FDR is controlled below the nominal level throughout. It equals $0.046$ at the level where the procedure first selects any unit ($q=0.75$) and reaches the trivial ``select-all'' floor of $0.50$ only as $q\to 1$. 

The demonstrated safety property does not require knowing that the policy is harmful for this subgroup a priori. Therefore, when the conformal policy selection method is applied to a subgroup where treatment helps, the same policy and gate would certify and deploy beneficial treatments. In this setting, where treatment harms, the gate refuses to treat them for small $q$. FDR control thus delivers safety for the subgroup. \bb{The parameter $q$ provides an interpretable budget for the expected false discovery proportion.}

\begin{figure}[!t]
    \centering
    \includegraphics[width=0.92\linewidth]{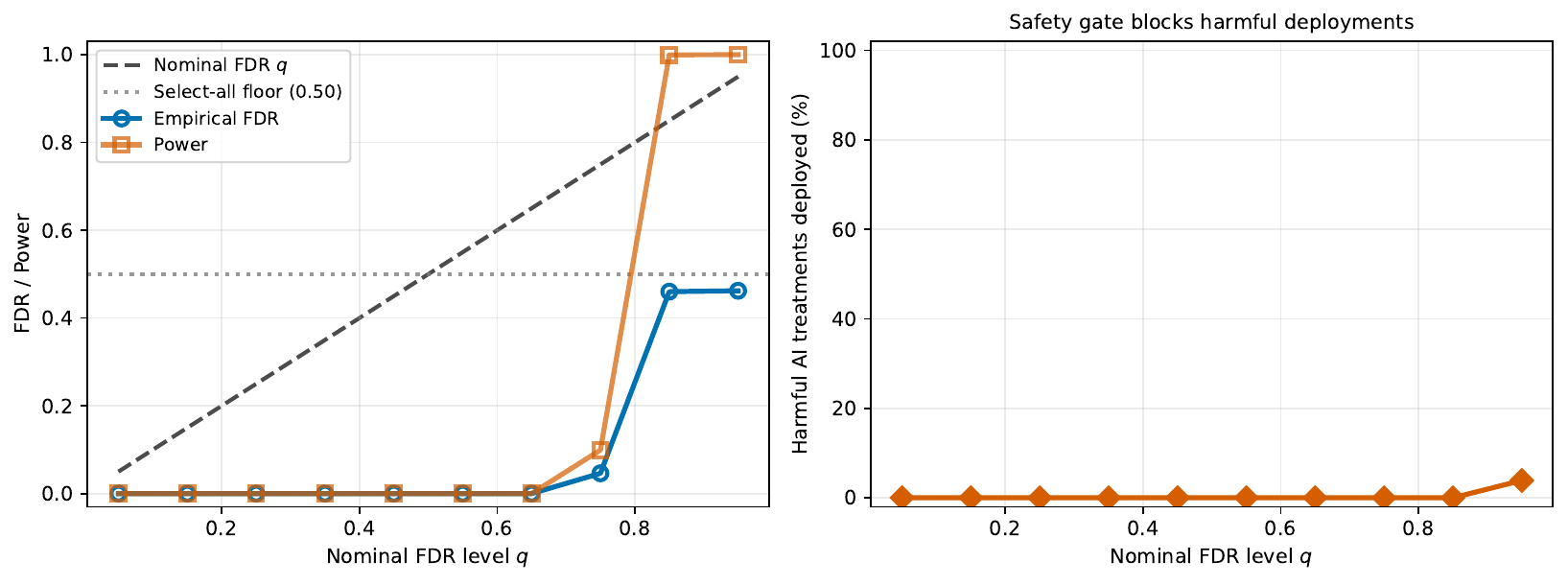}
    \caption{Conformal selection as a deployment safety gate on the EHR diabetes cohort (age $\ge 60$, BMI $\ge 30$, $n=5{,}993$, semi-synthetic causal-forest ground truth, $20$ replications). \textbf{Left}: the empirical FDR (blue) stays below the nominal level $q$ (dashed) across the whole range, reaching the select-all floor ($0.50$, dotted) only as $q\to1$, with power shown in orange. \textbf{Right}: the percentage of the AI policy's intended (harmful) treatments that are actually deployed through the gate, which is zero for all $q \le 0.85$ even though the gate certifies $93\%$ of the batch for deployment at $q=0.85$.}
    \label{fig:diabetes}
\end{figure}

\section{Discussion}

We introduced \emph{Conformal Policy Selection}, a framework that formulates targeted AI deployment as a multiple-testing problem with finite-sample FDR control over selected units. By pairing weighted conformal $p$-values with a two-step selection-and-pruning procedure, the method provides both a finite-sample guarantee (\Cref{thm:calib_ite}) and a \bb{doubly robust} asymptotic guarantee (\Cref{thm: asymptotic dr}) under which validity persists whenever \emph{either} the outcome model or the weight model is correctly specified. Rather than deploying a policy universally and hoping that the average gain does not introduce individual harm, a decision-maker can act only on the selected units with an explicit bound on incorrect selections.

An interesting feature of our formulation is that the selection target is the \emph{realized outcome under the deployed policy}, not the individual causal effect. \bb{We ask whether} a unit's outcome under $\pi$ clears the threshold $c$, which is arguably the operationally relevant quantity in most high-stakes settings, e.g.,  enrolling a patient, extending a loan, or allocating a scarce resource. Causal reasoning enters through the policy training stage and is deliberately separated from the selection stage, where we provide \bb{a distribution-free FDR guarantee}. When the causal effect itself is the quantity of interest, the same method applies by redefining the nonconformity score accordingly, a natural extension for future work.
\bibliographystyle{plainnat}
\bibliography{references}

\clearpage
\appendix

\section{Proofs}\label{app:proofs}

\subsection{Notations and preliminary results.}

We use the following notation throughout the proofs.

For $i=1,\ldots,n$, let $V_i = V(X_i,T_i,Y_i)$. For $j=1,\ldots,m$, let
$\widehat{V}_{n+j} = V(X_{n+j},T_{n+j},c_{n+j})$ and
$V_{n+j} = V(X_{n+j},T_{n+j},Y_{n+j})$. Also define
\[
  W_j := \sum_{i=1}^n w(X_i,T_i)+w(X_{n+j},T_{n+j}).
\]
Write $Z_{n+j}=(X_{n+j},T_{n+j},Y_{n+j})$ for the full test tuple and
$\widetilde{Z}_{n+j}=(X_{n+j},T_{n+j},c_{n+j})$ for its observable surrogate.
The true null set is
\[
  \cH_0 = \{j \in \cD_\test : Y_{n+j} \leq c_{n+j}\}.
\]
A unit $j$ is a \emph{false discovery} if $j\in\cS$ and $j\in\cH_0$.

Recall that the weighted conformal p-values are
\begin{equation}\label{eq:pj-cps}
  p_j \;=\;
  \frac{\displaystyle\sum_{i=1}^n w(X_i,T_i)\,\one\{V_i<\widehat{V}_{n+j}\}
        + w(X_{n+j},T_{n+j})}{W_j}.
\end{equation}

The oracle p-values, used only in the analysis, are defined as
\begin{equation}\label{eq:pj-cps-oracle}
  \bar{p}_j \;=\;
  \frac{\displaystyle\sum_{i=1}^n w(X_i,T_i)\,\one\{V_i<V_{n+j}\}
        + w(X_{n+j},T_{n+j})}{W_j}.
\end{equation}
That is, $\widehat{V}_{n+j}$ is replaced by the oracle score $V_{n+j}$.

On the null event $\{Y_{n+j}\leq c_{n+j}\}$, monotonicity of $V(x,t,\cdot)$
gives $V_{n+j}\leq\widehat{V}_{n+j}$, hence $\bar{p}_j\leq p_j$.

For theoretical purposes, define the randomized oracle p-value using an
independent $U_j\sim\mathrm{Unif}[0,1]$:
\begin{align}\label{eq:pjstar-cps}
  p_j^* \;=\;  \frac{1}{W_j}\!(
    \sum_{i=1}^n w(X_i,T_i)\,\one\{V_i<V_{n+j}\} +
    U_j\!(w(X_{n+j},T_{n+j})
    )).
\end{align}
By construction $p_j^*\leq\bar{p}_j$ always.

For $\ell\neq j$, define the auxiliary p-values
\begin{equation}\label{eq:aux-cps}
  \begin{aligned}
    p_\ell^{(j)}
    &:= \frac{1}{W_j}\Bigl(
      \sum_{i=1}^n w(X_i,T_i)\,\one\{V_i<\widehat{V}_{n+\ell}\} \\
    &\hspace{4.2em}
      + w(X_{n+j},T_{n+j})\,\one\{\widehat{V}_{n+j}<\widehat{V}_{n+\ell}\}
    \Bigr),
  \end{aligned}
\end{equation}
and set $p_j^{(j)}:=0$. The oracle auxiliary p-value uses $V_{n+j}$ instead of
$\widehat{V}_{n+j}$:
\begin{equation}\label{eq:aux-cps-oracle}
  \begin{aligned}
    \bar{p}_\ell^{(j)}
    &:= \frac{1}{W_j}\Bigl(
      \sum_{i=1}^n w(X_i,T_i)\,\one\{V_i<\widehat{V}_{n+\ell}\} \\
    &\hspace{4.2em}
      + w(X_{n+j},T_{n+j})\,\one\{V_{n+j}<\widehat{V}_{n+\ell}\}
    \Bigr),
  \end{aligned}
\end{equation}
for $\ell\neq j$.

Define $\widehat{\cR}_{j\to0}$ as the BH rejection set at level $q$ applied to
\[
  (p_1^{(j)},\ldots,p_{j-1}^{(j)},0,p_{j+1}^{(j)},\ldots,p_m^{(j)}).
\]
Set $K_j:=|\widehat{\cR}_{j\to0}|$ and $s_j:=qK_j/m$. Likewise, let
$\cR_{j\to0}$ denote the BH rejection set at level $q$ applied to
\[
  \bigl(\bar{p}_1^{(j)},\ldots,\bar{p}_{j-1}^{(j)},0,\bar{p}_{j+1}^{(j)},\ldots,\bar{p}_m^{(j)}\bigr).
\]
We also define
\[
  \widehat{\cR}_j
  := \mathrm{BH}(q)\text{ applied to }
  (p_1^{(j)},\ldots,p_{j-1}^{(j)},p_j,p_{j+1}^{(j)},\ldots,p_m^{(j)}),
\]
that is, the rejection set with the observable p-value $p_j$ in position $j$.
Similarly, let
\[
  \cR_j
  := \mathrm{BH}(q)\text{ applied to }
  (\bar{p}_1^{(j)},\ldots,\bar{p}_{j-1}^{(j)},\bar{p}_j,\bar{p}_{j+1}^{(j)},\ldots,\bar{p}_m^{(j)}).
\]

By the BH step-up property,
\begin{equation}\label{eq:hatRj-char}
  j\in\widehat{\cR}_j \;\iff\; p_j\leq s_j = \frac{q|\widehat{\cR}_{j\to0}|}{m},
\end{equation}
since placing $p_j$ in position $j$ (instead of $0$) either keeps $j$ in
the rejection set (if $p_j\leq s_j$) or removes it.

Finally, define
\[
  \cE_{z,j}=[Z_1,\ldots,Z_n,Z_{n+j}],
\]
the unordered multiset of the $n+1$ observations associated with calibration
points and test unit $j$.

The numerator of $\bar{p}_\ell^{(j)}$ can be written as a sum over all
elements $z=(x_z,t_z,y_z)$ of $\cE_{z,j}$:
\begin{align} \label{eq:aux-sum-cps}
  &\sum_{i=1}^n w(X_i,T_i)\one\{V_i<\widehat{V}_{n+\ell}\}
  + w(X_{n+j},T_{n+j})\one\{V_{n+j}<\widehat{V}_{n+\ell}\} \notag \\
  &= \sum_{\substack{z\in\cE_{z,j}}} w(x_z,t_z)\one\{V(z)<\widehat{V}_{n+\ell}\},
\end{align}
i.e., the sum treats $Z_{n+j}$ as one of the $n+1$ points on equal footing.
\eqref{eq:aux-sum-cps} depends on the \emph{unordered} set
$\cE_{z,j}$ and $\widehat{V}_{n+\ell}$ only, not on the labeling of which
draw is $Z_{n+j}$.

\subsection{Lemmas}

The following lemma shows that the oracle p-value $\bar{p}_j$ is super-uniform.

\begin{lemma}[Weighted exchangeability]  \label{lem:wexch-cps}
  Under~\eqref{eq:shift-xt}, conditional on $\cE_{z,j}$:
  \begin{enumerate}[label=(\roman*),noitemsep]
    \item $\displaystyle\PP(Z_{n+j}=z_i\mid\cE_{z,j})
      =\frac{w(x_i,t_i)}{W_j}$, $i\in\{1,\ldots,n,n+j\}$.
    \item $p_j^*\mid\cE_{z,j}\sim\mathrm{Unif}[0,1]$.
    \item $F_{z,j}(t):=\PP(\bar{p}_j\leq t\mid\cE_{z,j})\leq t$
          for all $t\in[0,1]$.
  \end{enumerate}
\end{lemma}

\begin{proof}
  \medskip\noindent\textbf{Part (i).}\enspace
  The joint density of $(Z_1,\ldots,Z_n,Z_{n+j})$ under $P^{\otimes n}\otimes Q$
  at realized values $(z_1,\ldots,z_n,z_{n+j})$ is
  \[
    \prod_{i=1}^n p(z_i)\cdot q(z_{n+j})
    = \prod_{i=1}^n p(z_i)\cdot w(x_{n+j},t_{n+j})\,p(z_{n+j}).
  \]
  For a fixed unordered multiset $\{z_1,\ldots,z_{n+j}\}$, the
  probability that $Z_{n+j}$ takes the value $z_{i_0}$ in this multiset
  is proportional to $w(x_{i_0},t_{i_0})\prod_{i=1}^{n+j}p(z_i)$.
  Since $\prod_i p(z_i)$ is the same for every choice of $i_0$,
  the conditional probability is
  \[
    \PP(Z_{n+j}=z_{i_0}\mid\cE_{z,j})
    = \frac{w(x_{i_0},t_{i_0})}{\sum_{i=1}^{n+j}w(x_i,t_i)}
    = \frac{w(x_{i_0},t_{i_0})}{W_j}.
  \]

  \medskip\noindent\textbf{Part (ii).}\enspace
  Given $\cE_{z,j}$, the scores $\{V(z_i)\}_{z_i\in\cE_{z,j}}$ form
  a fixed multiset.  Let $w_1<\cdots<w_K$ be the distinct score values
  with total weights $B_k=\sum_{z_i\in\cE_{z,j}:V(z_i)=w_k}w(x_i,t_i)$
  and cumulative weights $A_k=\sum_{\ell<k}B_\ell$.
  By part~(i), $V_{n+j}=V(Z_{n+j})=w_k$ with probability $B_k/W_j$.
  Given $V_{n+j}=w_k$ and $U_j=u$:
  \[
    p_j^* = \frac{A_k + u\,B_k}{W_j},
  \]
  which is $\mathrm{Unif}[A_k/W_j,\,(A_k+B_k)/W_j]$ since $U_j$ is
  uniform and independent of $\cE_{z,j}$.  The blocks
  $[A_k/W_j,(A_k+B_k)/W_j]$ partition $[0,1]$
  (since $\sum_k B_k=W_j$). 

Since $U_j\sim\mathrm{Unif}[0,1]$ is independent of $\cE_{z,j}$, conditional
  on $V_{n+j}=w_k$ the variable $p_j^{*}$ is uniform on
  $[A_k/W_j,\;(A_k+B_k)/W_j]$.

  For any $t\in[0,1]$, by the law of total probability:
  \begin{align*}
    \PP(p_j^{*}\leq t\mid\cE_{z,j})
    &= \sum_{k=1}^K \frac{B_k}{W_j}
       \cdot\PP\!\left(\frac{A_k+U_jB_k}{W_j}\leq t\;\Big|\; V_{n+j}=w_k\right).
  \end{align*}
  Let $t_k=A_k/W_j$ and $t_{k+1}=(A_k+B_k)/W_j$ be the left and right
  endpoints of the $k$-th block.  Then:
  \[
    \PP\!\left(\frac{A_k+U_jB_k}{W_j}\leq t\;\Big|\;V_{n+j}=w_k\right)
    = \min\!\left(\frac{(t-t_k)W_j}{B_k},1\right)^+,
  \]
  where $(x)^+=\max(x,0)$.  Substituting and summing the telescoping series
  (the blocks $[t_k,t_{k+1}]$ partition $[0,W_j/W_j]=[0,1]$ by construction),
  one obtains $\PP(p_j^{*}\leq t\mid\cE_{z,j})=t$ for every $t\in[0,1]$,
  i.e., $p_j^{*}\mid\cE_{z,j}\sim\mathrm{Unif}[0,1]$.

  \medskip\noindent\textbf{Part (iii): Stochastic dominance.}\enspace
  Since $p_j^{*}\leq\bar{p}_j$ always (the tie-breaking $U_j$ randomizes
  ties downward), we have for every $t$:
  \[
    F_{z,j}(t) = \PP(\bar{p}_j\leq t\mid\cE_{z,j})
    \leq \PP(p_j^{*}\leq t\mid\cE_{z,j}) = t. \qedhere
  \]
\end{proof}

Next, we reduce the FDR bound to a per-unit expression.

\begin{lemma}[FDR reduction]\label{lem:fdr-reduction}
  The CPS selection set $\mathcal{S}$ satisfies
$$
  \mathrm{FDR}
  \;\leq\;
  \sum_{j=1}^m\bE\!\left[\frac{\one\{p_j\leq s_j,\;Y_{n+j}\leq c_{n+j}\}}{K_j}\right].
$$
\end{lemma}

\begin{proof}

Recall the CPS selection set $\mathcal{S}=\{j:p_j\leq s_j,\;u\cdot K_j\leq r^*\}$
where $K_j=|\widehat{\cR}_{j\to0}|$, $s_j=qK_j/m$, $r^*=|\mathcal{S}|$,
and $u\sim\mathrm{Unif}(0,1)$ is independent of all data.
We reduce the FDR to a per-unit sum.
Condition on all data $\mathcal{D}$, leaving $u$ as the only
randomness.  Write $\mathcal{S}(u)$ to make the $u$-dependence explicit,
and let $r(u):=|\mathcal{S}(u)|$.

\emph{Monotonicity.} $r(u)$ is nonincreasing in $u$: increasing $u$
makes the condition $u\cdot K_\ell\leq r^*$ harder to satisfy,
so $r(u')\leq r(u)$ for $u'>u$.

\emph{Per-unit bound.} Fix any index $j$ satisfying
$p_j\leq s_j$ and $Y_{n+j}\leq c_{n+j}$.
Unit $j$ is selected iff $u K_j\leq r(u)$, i.e., iff $u\leq r(u)/K_j$.
By monotonicity, the event $\{u\leq r(u)/K_j\}$ equals $[0,u^*]$
where $u^*=\sup\{u:u K_j\leq r(u)\}$.
For every $u\in[0,u^*]$: $r(u)\geq r(u^*)\geq u^* K_j$, so
$1/r(u)\leq 1/(u^* K_j)$.  Hence:
\[
  \bE_u\!\left[\frac{\one\{j\in\mathcal{S}(u)\}}{r(u)}\right]
  =\int_0^{u^*}\frac{du}{r(u)}
  \leq\int_0^{u^*}\frac{du}{u^* K_j}
  =\frac{1}{K_j}.
\]

Marginalizing over $\mathcal{D}$ and summing:
\begin{equation}
  \mathrm{FDR}
  \;\leq\;
  \sum_{j=1}^m\bE\!\left[\frac{\one\{p_j\leq s_j,\;Y_{n+j}\leq c_{n+j}\}}{K_j}\right].
\end{equation}
\end{proof}

\subsection{Proof of \texorpdfstring{\Cref{thm:calib_ite}}{Theorem}}

\begin{proof}[Proof of \Cref{thm:calib_ite}.]

Fix $j\in\{1,\ldots,m\}$. By \Cref{lem:fdr-reduction}, it suffices to prove
\[
  \bE\!\left[\frac{\one\{p_j\leq s_j,\;Y_{n+j}\leq c_{n+j}\}}{K_j}\right]
  \leq \frac{q}{m}.
\]

We first compare the observable and oracle quantities on the event
\[
  E_j := \{Y_{n+j}\leq c_{n+j}\}\cap\{p_j\leq s_j\}.
\]

\medskip\noindent\emph{Step 1: oracle substitution.}
On $\{Y_{n+j}\leq c_{n+j}\}$, monotonicity of $V(x,t,\cdot)$ gives
$V_{n+j}\leq\widehat{V}_{n+j}$. Hence
$\one\{V_i<V_{n+j}\}\leq\one\{V_i<\widehat{V}_{n+j}\}$ for each $i$, so
\[
  \bar{p}_j\leq p_j\leq s_j=\frac{qK_j}{m}.
\]

For each $\ell\neq j$, the only difference between $\bar{p}_\ell^{(j)}$ and
$p_\ell^{(j)}$ is whether the $j$-th contribution is
$w(X_{n+j},T_{n+j})\,\one\{V_{n+j}<\widehat{V}_{n+\ell}\}$ or
$w(X_{n+j},T_{n+j})\,\one\{\widehat{V}_{n+j}<\widehat{V}_{n+\ell}\}$.
Since $V_{n+j}\leq\widehat{V}_{n+j}$, three cases are possible.

\emph{Case A} ($\widehat{V}_{n+\ell}>\widehat{V}_{n+j}$): both indicators are
$1$, so $\bar{p}_\ell^{(j)}=p_\ell^{(j)}$.

\emph{Case B} ($\widehat{V}_{n+\ell}\leq V_{n+j}$): both indicators are $0$,
so $\bar{p}_\ell^{(j)}=p_\ell^{(j)}$.

\emph{Case C} ($V_{n+j}<\widehat{V}_{n+\ell}\leq\widehat{V}_{n+j}$):
$\bar{p}_\ell^{(j)}>p_\ell^{(j)}$, but both p-values are still at most
$qK_j/m$. Indeed,
\begin{align*}
  p_\ell^{(j)}
  &= \frac{\sum_{i=1}^n w(X_i,T_i)\,\one\{V_i<\widehat{V}_{n+\ell}\}}{W_j}
  \leq \frac{\sum_{i=1}^n w(X_i,T_i)\,\one\{V_i<\widehat{V}_{n+j}\}}{W_j} \\
  &\leq p_j
  \leq \frac{qK_j}{m}, \\
  \bar{p}_\ell^{(j)}
  &= \frac{1}{W_j}\Bigl(
       \sum_{i=1}^n w(X_i,T_i)\,\one\{V_i<\widehat{V}_{n+\ell}\} \\
  &\hspace{4.2em}
       + w(X_{n+j},T_{n+j})
     \Bigr)
  \leq p_j
  \leq \frac{qK_j}{m}.
\end{align*}

Now compare the BH runs that define $\cR_{j\to0}$ and
$\widehat{\cR}_{j\to0}$. Both use $0$ in position $j$. For $\ell\neq j$,
Cases A and B give $\bar{p}_\ell^{(j)}=p_\ell^{(j)}$, while Case C gives
$p_\ell^{(j)}\leq\bar{p}_\ell^{(j)}\leq qK_j/m$. Thus $\ell$ crosses the
threshold $qK_j/m$ in both vectors whenever it does in either one. By the
self-consistency of the BH step-up rule,
\[
  \cR_{j\to0}=\widehat{\cR}_{j\to0}
  \qquad\text{and hence}\qquad
  |\cR_{j\to0}|=|\widehat{\cR}_{j\to0}|=K_j.
\]

Consequently, on $E_j$,
\[
  \bar{p}_j\leq p_j\leq s_j=\frac{qK_j}{m}
  = \frac{q|\cR_{j\to0}|}{m},
\]
and therefore
\begin{equation}\label{eq:ptwise-cps}
  \frac{\one\{p_j\leq s_j,\;Y_{n+j}\leq c_{n+j}\}}{K_j}
  \;\leq\;
  \frac{\one\!\bigl\{\bar{p}_j\leq q|\cR_{j\to0}|/m\bigr\}}{|\cR_{j\to0}|}.
\end{equation}
On $E_j^c$, the left-hand side is $0$, so \eqref{eq:ptwise-cps} holds
everywhere.

\medskip\noindent\emph{Step 2: conditional independence.}
We next show that $|\cR_{j\to0}|$ and $\bar{p}_j$ are independent
conditional on $\cE_{z,j}$.

For $\ell\neq j$, \eqref{eq:aux-sum-cps} shows that
\[
  \bar{p}_\ell^{(j)}
  = \frac{\sum_{z\in\cE_{z,j}} w(x_z,t_z)\,\one\{V(z)<\widehat{V}_{n+\ell}\}}
         {W_j},
\]
so $\bar{p}_\ell^{(j)}$ depends only on the unordered multiset $\cE_{z,j}$ and
on $\widehat{V}_{n+\ell}=V(X_{n+\ell},T_{n+\ell},c_{n+\ell})$. Hence
$|\cR_{j\to0}|$ is a function of $\cE_{z,j}$ and
$\{\widetilde{Z}_{n+\ell}\}_{\ell\neq j}$.

By contrast,
\[
  \bar{p}_j = \frac{\sum_{i=1}^n w(X_i,T_i)\,\one\{V_i<V_{n+j}\}
                   + w(X_{n+j},T_{n+j})}{W_j}
\]
depends on which element of $\cE_{z,j}$ is labeled $Z_{n+j}$, equivalently on
$Y_{n+j}$. By assumption~(iii) of Theorem~\ref{thm:calib_ite},
$Y_{n+j}$ is independent of $\{\widetilde{Z}_{n+\ell}\}_{\ell\neq j}$ given
$\cE_{z,j}$. Therefore $\bar{p}_j$ is independent of $|\cR_{j\to0}|$ given
$\cE_{z,j}$.

Conditioning on $(\cE_{z,j},\{\widetilde{Z}_{n+\ell}\}_{\ell\neq j})$ makes
$|\cR_{j\to0}|=k$ a constant, and
\begin{align*}
  \PP(\bar{p}_j\leq qk/m \mid \cE_{z,j},\{\widetilde{Z}_{n+\ell}\}_{\ell\neq j})
  &= \PP(\bar{p}_j\leq qk/m \mid \cE_{z,j}) \\
  &= F_{z,j}(qk/m) \\
  &\leq \frac{qk}{m}
\end{align*}
by Lemma~\ref{lem:wexch-cps}(iii). Using the conditional independence,
\begin{align}
  \bE\!\left[
    \frac{\one\{\bar{p}_j\leq q|\cR_{j\to0}|/m\}}{|\cR_{j\to0}|}
    \,\Bigg|\, \cE_{z,j},
    \{\widetilde{Z}_{n+\ell}\}_{\ell\neq j}
  \right]
  &= \frac{\PP(\bar{p}_j\leq qk/m\mid\cE_{z,j})}{k} \notag \\
  &\leq \frac{q}{m}, \label{eq:tower-cps}
\end{align}
where $k=|\cR_{j\to0}|$ under the conditioning.

Finally,
\begin{align}
  \bE\!\left[\frac{\one\{p_j\leq s_j,\;Y_{n+j}\leq c_{n+j}\}}{K_j}\right]
  &\leq \bE\!\left[\frac{\one\{\bar{p}_j\leq q|\cR_{j\to0}|/m\}}
                         {|\cR_{j\to0}|}\right] \notag \\
  &= \bE\!\left[
       \bE\!\left[
         \frac{\one\{\bar{p}_j\leq q|\cR_{j\to0}|/m\}}{|\cR_{j\to0}|}
         \,\Bigg|\, \cE_{z,j},
         \{\widetilde{Z}_{n+\ell}\}_{\ell\neq j}
       \right]
     \right] \notag \\
  &\leq \frac{q}{m}, \label{eq:tag}
\end{align}
where the first inequality is \eqref{eq:ptwise-cps}, the second step is the
tower property, and the last inequality is \eqref{eq:tower-cps}.

Substituting the per-unit bound into \Cref{lem:fdr-reduction}:
\[
  \mathrm{FDR}
  \;\leq\;
  \sum_{j=1}^m\bE\!\left[\frac{\one\{p_j\leq s_j,\;Y_{n+j}\leq c_{n+j}\}}{K_j}\right]
  \;\leq\;
  \sum_{j=1}^m\frac{q}{m}
  \;=\; q. \qedhere
\]

\end{proof}

\subsection{Proof of \texorpdfstring{\Cref{thm: asymptotic dr}}{Theorem}}
\begin{proof}[Proof of \Cref{thm: asymptotic dr}.]
Throughout the proof, interpret the quantities
$p_j$, $\bar{p}_j$, $p_\ell^{(j)}$, and $\bar{p}_\ell^{(j)}$
as being constructed with the working weight $w^\dagger$.

The oracle-substitution argument from the proof of
\Cref{thm:calib_ite} uses only the monotonicity of the score and the
self-consistency of the BH step-up rule. Hence the pointwise inequality
\eqref{eq:ptwise-cps} continues to hold:
\[
  \frac{\one\{p_j\leq s_j,\;Y_{n+j}\leq c_{n+j}\}}{K_j}
  \;\leq\;
  \frac{\one\!\bigl\{\bar{p}_j\leq q|\cR_{j\to0}|/m\bigr\}}{|\cR_{j\to0}|}.
\]
Taking expectations and applying \Cref{lem:fdr-reduction} gives
\begin{equation}\label{eq:fdr-oracle-dr}
  \mathrm{FDR}
  \;\leq\;
  \sum_{j=1}^m
  \bE\!\left[\frac{\one\!\bigl\{\bar{p}_j\leq q|\cR_{j\to0}|/m\bigr\}}{|\cR_{j\to0}|}\right].
\end{equation}

Let $F_V$ denote the marginal CDF of the residual score $V=\epsilon$.
Because $V_i=\epsilon_i$ is independent of $(X_i,T_i)$, the strong law of
large numbers yields, for each fixed $j$,
\[
  \frac{1}{n}\sum_{i=1}^n w^\dagger(X_i,T_i)\,\one\{V_i<V_{n+j}\}
  \;\xrightarrow{\;a.s.\;}
  \bE\!\left[w^\dagger(X,T)\right] F_V(V_{n+j}),
\]
while
\[
  \frac{W_j}{n}
  = \frac{1}{n}\sum_{i=1}^n w^\dagger(X_i,T_i)
    + \frac{w^\dagger(X_{n+j},T_{n+j})}{n}
  \;\xrightarrow{\;a.s.\;}
  \bE\!\left[w^\dagger(X,T)\right].
\]
The bounded-ratio assumption implies
$w^\dagger(X_{n+j},T_{n+j})/W_j\leq \bar{w}/(n+1)\to 0$, so
\begin{equation}\label{eq:slln-dr}
  \bar{p}_j \;\xrightarrow{\;a.s.\;}\; F_V(V_{n+j}).
\end{equation}
The same argument, with $V_{n+j}$ replaced by $\widehat{V}_{n+\ell}$, gives
\[
  \bar{p}_\ell^{(j)} \;\xrightarrow{\;a.s.\;}\;
  F_V\!\bigl(\widehat{V}_{n+\ell}\bigr), \qquad \ell\neq j.
\]

Since $|\cR_{j\to0}|$ is a deterministic function of the finite collection
$\{\bar{p}_\ell^{(j)}:\ell\neq j\}$, it converges almost surely to a limit
$k_j^\infty\in\{1,\ldots,m\}$. Because the integrand in
\eqref{eq:fdr-oracle-dr} is bounded by $1$, the bounded convergence theorem
gives
\[
  \lim_{n\to\infty}
  \bE\!\left[\frac{\one\!\bigl\{\bar{p}_j\leq q|\cR_{j\to0}|/m\bigr\}}
                  {|\cR_{j\to0}|}\right]
  =
  \bE\!\left[\frac{\one\!\bigl\{F_V(V_{n+j})\leq qk_j^\infty/m\bigr\}}
                  {k_j^\infty}\right].
\]

By assumption~(iii), $V_{n+j}=\epsilon_{n+j}$ is independent of
$\{\widetilde{Z}_{n+\ell}\}_{\ell\neq j}$, and therefore independent of
$k_j^\infty$. Moreover, $F_V(V_{n+j})$ is super-uniform:
$\PP(F_V(V_{n+j})\leq t)\leq t$ for every $t\in[0,1]$. Conditioning on
$k_j^\infty$ therefore gives
\[
  \bE\!\left[\frac{\one\!\bigl\{F_V(V_{n+j})\leq qk_j^\infty/m\bigr\}}
                  {k_j^\infty}\right]
  \leq
  \bE\!\left[\frac{qk_j^\infty/m}{k_j^\infty}\right]
  = \frac{q}{m}.
\]
Hence
\[
  \limsup_{n\to\infty}
  \bE\!\left[\frac{\one\!\bigl\{\bar{p}_j\leq q|\cR_{j\to0}|/m\bigr\}}
                  {|\cR_{j\to0}|}\right]
  \leq \frac{q}{m}.
\]
Summing \eqref{eq:fdr-oracle-dr} over $j=1,\ldots,m$ yields
\[
  \limsup_{n\to\infty}\mathrm{FDR}
  \;\leq\;
  \sum_{j=1}^m \frac{q}{m}
  \;=\; q. \qedhere
\]
\end{proof}

\subsection{Proof of \texorpdfstring{\Cref{thm:estimated}}{Theorem}}
\begin{proof}[Proof of \Cref{thm:estimated}.]
Let $\cS^N$ denote the selection set produced by
Algorithm~\ref{alg:bh} using the estimated nuisance functions
$(\hat{f}_N,\hat{e}_N)$, and let $\cS$ denote the corresponding ideal
selection set obtained with the limit functions $f$ and $e$.

Define
\begin{align*}
  V_i^N &:= Y_i-\hat{f}_N(X_i,T_i), \\
  \hat{V}_{n+j}^N &:= c_{n+j}-\hat{f}_N(X_{n+j},T_{n+j}), \\
  \hat{w}_N(x,t) &:= \frac{1}{\hat{e}_N(x)}.
\end{align*}
By Assumption~\ref{assump: weight and estimator}, conditioning on the
independent training sample and letting $N\to\infty$, we have almost surely
\begin{align*}
  V_i^N &\to V_i := Y_i-f(X_i,T_i), \\
  \hat{V}_{n+j}^N &\to \hat{V}_{n+j} := c_{n+j}-f(X_{n+j},T_{n+j}), \\
  \hat{w}_N(x,t) &\to w(x,t) := \frac{1}{e(x)}.
\end{align*}
Consequently, every observable conformal p-value $p_j^N$ and every auxiliary
p-value $(p_\ell^{(j)})^N$ converges almost surely to its ideal counterpart.
Since Algorithm~\ref{alg:bh} applies BH and then a finite pruning rule to
finitely many p-values, the corresponding rejection counts $K_j^N$ and the
selection set $\cS^N$ converge almost surely to $K_j$ and $\cS$, except
possibly on the measure-zero event of exact ties, which standard randomized
tie-breaking resolves if needed.

For the limiting procedure, the score $V(x,t,y)=y-f(x,t)$ is nondecreasing in
$y$, the weight $w(x,t)=1/e(x)$ satisfies the shift condition
\eqref{eq:shift-xt}, and sample splitting preserves condition~(iii) of
Theorem~\ref{thm:calib_ite}. Therefore
Theorem~\ref{thm:calib_ite} yields
\[
  \mathrm{FDR}(\cS)\leq q.
\]

Since $\cS^N\to\cS$ almost surely, the false discovery proportion satisfies
$\mathrm{FDP}(\cS^N)\to\mathrm{FDP}(\cS)$ almost surely. As
$0\leq \mathrm{FDP}\leq 1$, the bounded convergence theorem gives
\[
  \lim_{N\to\infty}\mathrm{FDR}(\cS^N)
  = \mathrm{FDR}(\cS)
  \leq q.
\]
Hence $\mathrm{FDR}\lesssim q$ as $N\to\infty$.
\end{proof}

\section{Materials and Methods}\label{app:methods}

\subsection{Data}
The synthetic experiments use Gaussian covariates $X\sim\mathcal{N}(0,I_5)$, a logistic behavior policy, and linear potential-outcome models, as detailed in the Numerical Experiments section. The diabetes-prevention analysis uses the longitudinal electronic health record dataset of \bb{\citet{kraus2024data}} ($89{,}999$ prediabetic patients, $540{,}419$ patient-year records, 2003--2013). We restrict to fully observed follow-up windows (focal year $\le 2010$) and to the high-risk subgroup (age $\ge 60$, BMI $\ge 30$), which yields $n=5{,}993$ patient-years. The LaLonde experiment uses the standard job-training benchmark of \bb{\citet{lalonde1986evaluating}} (614 individuals, of whom 185 are treated and 429 are control).

\subsection{Models and procedure}
For each dataset we split the data into training, calibration, and test folds (stratified by treatment) and fit three models on the training fold: an outcome model $\hat\mu(x,t)$ defining the nonconformity score $V=\tilde Y-\hat\mu(X,T)$ (random forest for the real-data experiments, the true linear model for the synthetic study), a policy model producing the deployment rule $\pi(x)=\mathbb{I}\{\hat\mu_{\mathrm{lin}}(x,1)>\hat\mu_{\mathrm{lin}}(x,0)\}$, and an $L_2$-regularized logistic propensity model $\hat e(x)$ (clipped to $[0.05,0.95]$) defining the inverse-propensity weights. We then run Algorithm~\ref{alg:bh} (weighted conformal selection with heterogeneous pruning) and its weighted-BH variant, sweeping the nominal FDR level $q\in[0.05,0.95]$ and averaging over independent splits.

\section{Real-World Data: LaLonde Job Training}
\label{sec:lalonde}

We evaluate our method on the LaLonde job training dataset \citep{lalonde1986evaluating}, a classic benchmark for causal inference. The dataset contains 614 individuals from a randomized job training program, with 185 treated and 429 controls. The covariates include age, education, race indicators, marital status, degree status, and prior earnings (1974 and 1975). The outcome is post-program earnings in 1978 (\texttt{re78}).

\subsubsection{Setup}
We perform 100 random stratified splits of the data into training (30\%), calibration (40\%), and test (30\%) sets, preserving the treatment ratio in each split.
On each training set, we fit a Random Forest outcome model $\hat\mu(x,t)$ and a logistic propensity model $\hat{e}(x) = \hat{\Pr}(T=1|X=x)$ with standardized features.
The AI policy assigns treatment where the estimated CATE is positive: $\pi(x) = \ind\{\hat\mu(x,1) > \hat\mu(x,0)\}$.
The nonconformity scores are $V_i = Y_i - \hat\mu(X_i, T_i)$, and the IPW weights use the estimated propensity clipped to $[0.05, 0.95]$.
The selection threshold is $c=4{,}000$, targeting individuals with annual earnings exceeding four thousand dollars under the AI policy. This threshold sits just below the median of \texttt{re78} in the dataset (\textdollar 4{,}976 for controls, \textdollar 4{,}232 for treated), separating ``clearly employed at a meaningful wage'' from ``marginally employed or unemployed,'' and is well-matched to the signal present in the covariates.

Since counterfactual outcomes are not observed in real data, we evaluate using a hybrid ground truth: for units where the observed treatment matches the AI recommendation ($T_i = \pi(X_i)$), we use the factual outcome $Y_i$, and otherwise we use the model prediction $\hat\mu(X_i, \pi(X_i))$ as a proxy.

\subsubsection{Results}
Figure~\ref{fig:lalonde} shows the FDR and power curves at $c=4{,}000$, averaged over 100 random splits. The empirical FDR is controlled at or below the nominal level $q$ throughout the range $q \in [0.05, 0.65]$ (e.g., FDR $\approx 0.025$ at $q=0.15$, FDR $\approx 0.125$ at $q=0.65$), and then plateaus at the baseline rate of test units with true outcome below the threshold (roughly $23\%$). Power rises smoothly from $0.012$ at $q=0.15$ to $0.34$ at $q=0.65$ and approaches $1$ as $q \to 0.85$. At $q=0.15$ the procedure selects on average $1.9$ individuals per test batch of $\sim 185$, and selected individuals have mean earnings of \textdollar 20{,}780 versus \textdollar 8{,}244 for non-selected (difference \textdollar 12{,}535), and the fraction clearing the threshold is $93.1\%$ among selected versus $76.7\%$ among non-selected ($+16.4$ percentage points). Both WBH and WCS with heterogeneous pruning perform similarly on this dataset, with WCS slightly more conservative at small $q$.

\begin{figure}[!t]
    \centering
    \begin{subfigure}[b]{0.48\linewidth}
        \centering
        \includegraphics[width=\linewidth]{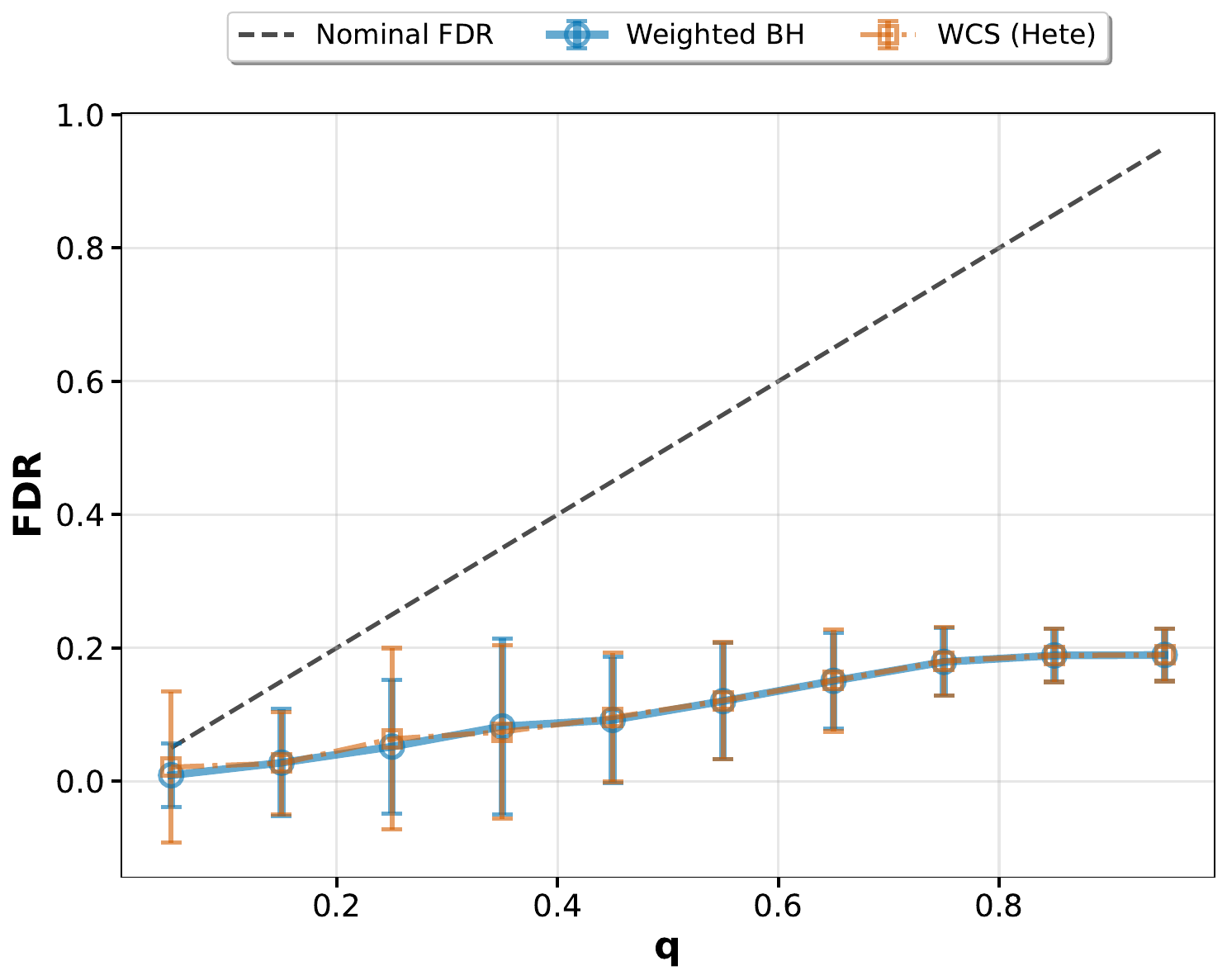}
    \end{subfigure}
    \hfill
    \begin{subfigure}[b]{0.48\linewidth}
        \centering
        \includegraphics[width=\linewidth]{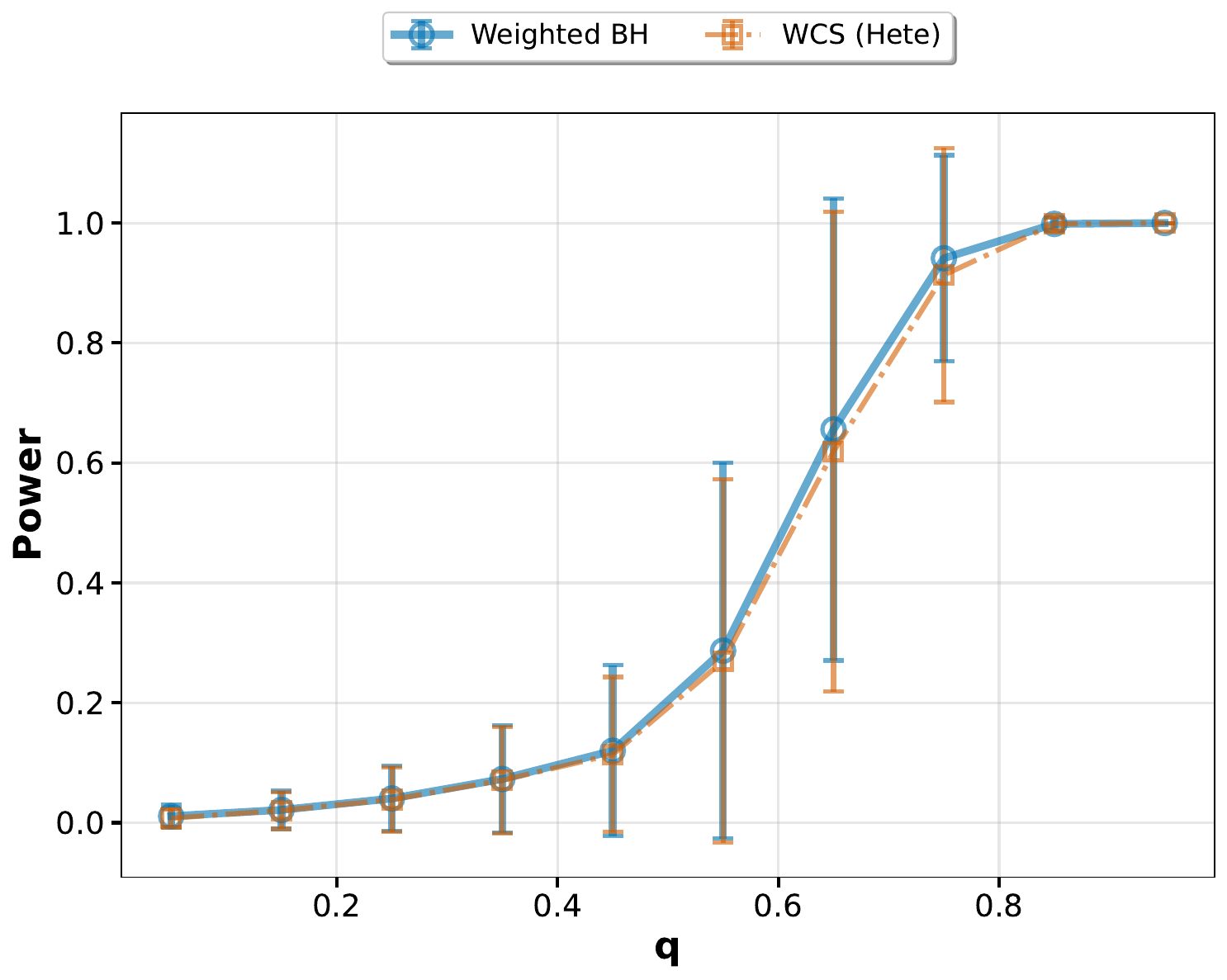}
    \end{subfigure}
    \caption{FDR and power on the LaLonde job training dataset (614 observations, 100 random splits) with threshold $c=4{,}000$. Empirical FDR stays at or below the nominal level $q$ for $q \le 0.65$ and plateaus near the baseline negative rate ($\approx 0.23$) as $q$ approaches $1$. Power increases smoothly with $q$, reflecting the trade-off the user controls by choosing the FDR budget.}
    \label{fig:lalonde}
\end{figure}

\Cref{fig:lalonde_summary} reports the mean earnings for the selected and non-selected groups. When $q$ is small, the average earnings of the selected group are substantially larger than those of the non-selected group. When $q$ is large (near $0.9$), the procedure selects almost everyone, so the mean earnings of the selected group decline toward the overall average. This demonstrates the benefit of selective treatment.

\begin{figure}[!t]
    \centering
    \includegraphics[width=0.9\linewidth]{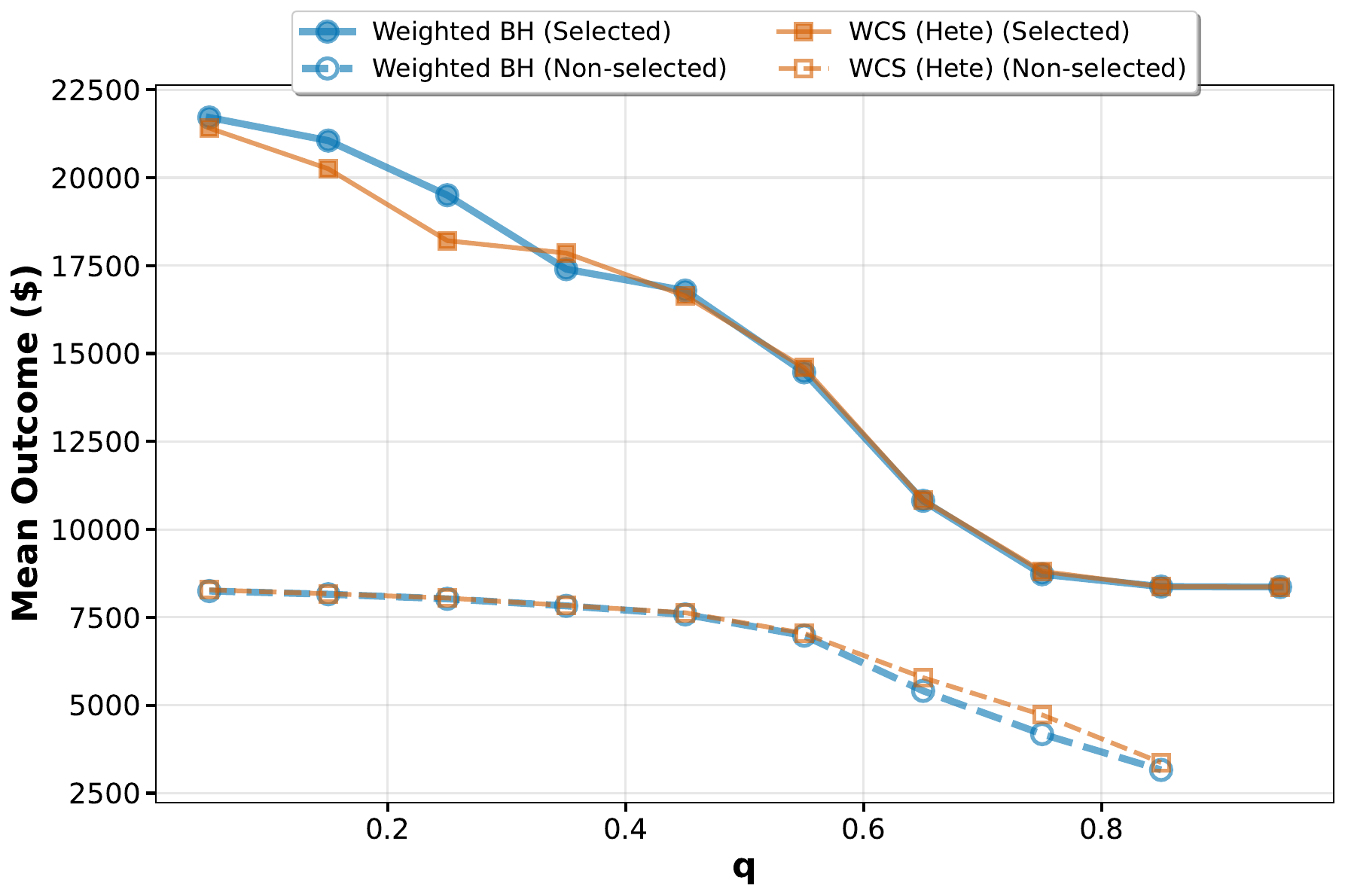}
    \caption{LaLonde job training dataset. Mean post-program earnings (\texttt{re78}) of the selected versus non-selected groups as a function of the nominal FDR level $q$. At small $q$ the selected group has markedly higher mean earnings, and as $q\to 0.9$ almost all units are selected and the two curves converge.}
    \label{fig:lalonde_summary}
\end{figure}

\section{Experimental Details of the Diabetes Experiment}
\label{sec:si-diabetes}

This appendix gives the full details in the diabetes
experiment (\Cref{fig:diabetes} of the main text). We focus on the
clinically high-risk subgroup of the EHR data of \citet{kraus2024data} (age
$\ge 60$, BMI $\ge 30$, $n=5{,}993$ patient-years, $305$ features, $8.3\%$
treated, $36.5\%$ three-year onset rate).

\subsection{Estimating Treatment Effect}
\label{sec:si-confounding}

In this cohort antidiabetics are prescribed reactively, so the treatment
indicator is strongly correlated with imminent onset. The raw three-year avoidance rate
is $3.0\%$ among treated versus $69.0\%$ among untreated. 
Table~\ref{tab:si-cate} reports the mean estimated treatment effect $\hat\tau(x)$ and the fraction of patients with a
\emph{positive} estimated effect for four standard heterogeneous-effect estimators. In every case the estimate is negative for
essentially all patients. Table~\ref{tab:si-overlap} stratifies by estimated
propensity decile and shows that even where treated and untreated patients
overlap (deciles with a balanced treated share), the treated avoid onset far
less often. 
Avoidance means not developing (progressing to) type-II diabetes within the three-year window.

\begin{table}[h]
  \centering
  \caption{Estimated treatment effect on three-year avoidance for four
  heterogeneous-effect estimators (age $\ge 60$, BMI $\ge 30$ subgroup). Every
  estimator returns a negative effect for essentially all patients.}
  \label{tab:si-cate}
  \begin{tabular}{lccc}
    \toprule
    Estimator & mean $\hat\tau$ & frac. $\hat\tau>0$ & max $\hat\tau$ \\
    \midrule
    S-learner (random forest)   & $-0.269$ & $0.0\%$ & $-0.052$ \\
    T-learner (random forest)   & $-0.981$ & $0.0\%$ & $+0.000$ \\
    X-learner (random forest)   & $-0.994$ & $0.0\%$ & $+0.000$ \\
    Causal Forest (DML)         & $-0.570$ & $0.0\%$ & $-0.304$ \\
    \bottomrule
  \end{tabular}
\end{table}

\begin{table}[h]
  \centering
  \caption{Three-year avoidance rate of treated vs.\ untreated patients within
  estimated-propensity deciles. Even in balanced (high-propensity) deciles the
  treated fare far worse, confirming the harm is not an overlap artifact.}
  \label{tab:si-overlap}
  \begin{tabular}{lcccc}
    \toprule
    Propensity decile & $n_{\text{treated}}$ & avoid & $n_{\text{control}}$ & avoid \\
    \midrule
    4  & $15$  & $13.3\%$ & $584$ & $72.1\%$ \\
    6  & $30$  & $3.3\%$  & $570$ & $68.4\%$ \\
    8  & $58$  & $5.2\%$  & $541$ & $65.4\%$ \\
    10 & $250$ & $1.6\%$  & $350$ & $57.7\%$ \\
    \bottomrule
  \end{tabular}
\end{table}

\subsection{Semi-synthetic ground truth}
\label{sec:si-dgp}

Because the counterfactual outcome $Y(\pi(X))$ is never observed, the true FDR
cannot be measured on the raw data (a factual/model-proxy hybrid inflates it).
We therefore construct a semi-synthetic data-generating process on the
\emph{real} covariates. A causal forest ($300$ trees,
minimum leaf $20$, with random-forest first stages and $3$-fold cross-fitting)
fit on the full subgroup provides the conditional treatment effect
$\tau(x)$. The control response curve 
$\mu_0(x)=\Pr(\tilde Y=1\mid x, T=0)$ is fit by a random forest on the untreated
units, and $\mu_1(x)=\mathrm{clip}(\mu_0(x)+\tau(x),\,0.001,\,0.999)$. These
oracle surfaces satisfy $\bar\mu_0=0.686$, $\bar\mu_1=0.118$, $\bar\tau=-0.568$, matching \S\ref{sec:si-confounding}. Continuous potential
outcomes are drawn per replication as
\[
  Y(t) = \mu_t(X) + \varepsilon,\qquad \varepsilon\sim\mathcal N(0,\sigma^2),\quad \sigma=0.10,
\]
and the observed outcome is $Y^{\mathrm{obs}}=T\,Y(1)+(1-T)\,Y(0)$ with the
factual treatment $T$. Because both potential outcomes are available, the realized outcome $Y(\pi(X))$ under any policy can be computed.

\subsection{AI policy and selection procedure}
\label{sec:si-policy}

The AI policy is the standard clinical rule of treating the population with the highest-risk. On each training fold we fit a random-forest onset-probability model
$\widehat{\Pr}(\text{onset}\mid x)$ and set
$\pi(x)=\mathbb{I}\{\widehat{\Pr}(\text{onset}\mid x)\ge \hat\kappa\}$, where
$\hat\kappa$ is the $90$th training percentile. This deploys treatment to
$\approx 7$--$10\%$ of the test units. The nonconformity score uses a
separate random-forest S-learner $\hat\mu(x,t)$. Calibration scores are
$V_i=Y^{\mathrm{obs}}_i-\hat\mu(X_i,T_i)$ and test scores are
$\widehat V_{n+j}=c-\hat\mu(X_{n+j},\pi(X_{n+j}))$, with $c$ the median of the
realized policy outcome $Y(\pi)$ over the test set. Policy-shift weights are
$w_i=\mathbb{I}\{T_i=\pi(X_i)\}/\hat\pi_0(T_i\mid X_i)$ with an $L_2$-regularized
logistic propensity clipped to $[0.05,0.95]$, as in \Cref{eq:weighted_pval}. We
draw $20$ random $30\%/60\%/10\%$ train/calibration/test splits stratified by
treatment and run weighted Benjamini--Hochberg selection on the weighted conformal
$p$-values over $q\in[0.05,0.95]$. 

\subsection{Safety metric}
\label{sec:si-safety}

The AI policy is deployed only to the selected set $\mathcal S$, and every other unit
receives the status quo. A \emph{harmful deployment} is a selected unit that the
policy would treat, $\{j\in\mathcal S:\pi(X_{n+j})=1\}$. Since treatment is
harmful throughout this subgroup, each such unit has a realized
outcome under the deployed policy that is degraded. Because these units have a low
predicted outcome under treatment, their test scores are large and they are
excluded from $\mathcal S$ until $q$ is near one. The right panel of
\Cref{fig:diabetes} reports the number of harmful deployments as a fraction of
the $\sim 44$ treatments per batch for the policy. It is zero for all
$q\le 0.85$ and reaches only $\sim 2$ units at $q=0.95$. Because FDR control
bounds this count of harmful deployments by $q\,|\mathcal S|$, controlling the
false-discovery rate makes the safety guarantee adaptive and no prior
knowledge that treatment is harmful here is required.

\section{Diabetes Experiment with Heterogeneous Treatment Effects}
\label{sec:si-hetero}

In \Cref{sec:si-safety}, treatment is harmful throughout the
subgroup, which makes the safety gate protective but degenerate since the ideal selected set is empty. 

To show that the gate discriminates when
the effect is genuinely heterogeneous, we build another semi-synthetic study
on the same real electronic-health-record covariates but inject a treatment
effect that is beneficial for some patients and harmful for others.

We keep the observed covariates and fit a control response curve for the outcome under no
treatment. We then inject a heterogeneous effect
$\tau(X)=a\,(z(X)-\delta)$, where $z(X)$ is a fixed standardized covariate
and $\delta$ is set so that roughly half of the patients have a positive effect.
The treated potential outcome is the control surface plus $\tau(X)$. To ensure
overlap, we simulate treatment from a
well-behaved behavior policy with propensities clipped to $[0.15,0.85]$. The
deployed AI policy is a T-learner estimate of the effect thresholded at zero, so
it treats the patients it predicts will benefit. 
A causal forest fit on the full
generated data serves as the ground-truth effect, and a treated unit is
counted as a harmful deployment when it is selected, and
has a true negative effect. We run $20$ repetitions with the same weighted
Benjamini--Hochberg selection as in \Cref{sec:si-safety}.

\begin{table}[t]
\centering
\caption{Heterogeneous safety gate on real covariates, averaged over $20$
repetitions. The policy treats about half of the cohort.
The realized FDR stays below the nominal $q$.}
\label{tab:si-hetero}
\begin{tabular}{lccccc}
\hline
$q$ & FDR & Power & Frac.\ treated & Beneficial treated & Harmful deployed \\
\hline
$\le 0.35$ & $0.00$ & $0.00$ & $0.00$ & $0$ & $0$ \\
$0.45$ & $0.22$ & $0.85$ & $0.49$ & $76.5$ & $24.4$ \\
$0.55$ & $0.29$ & $1.00$ & $0.49$ & $91.4$ & $30.5$ \\
$\ge 0.65$ & $0.29$ & $1.00$ & $0.49$ & $91.4$ & $30.5$ \\
\hline
\end{tabular}
\end{table}

\Cref{tab:si-hetero} summarizes the result. The AI policy 
treats about $49\%$ of the cohort. The realized FDR  never exceeds the target $q$. At
every level the selection results in far more beneficial treated units than harmful
deployments. When the effect is heterogeneous the conformal selection lets
the beneficial treatments pass while bounding the harmful ones. \Cref{fig:si-hetero} shows the corresponding FDR,
power, and safety curves.

\begin{figure}[t]
\centering
\includegraphics[width=\linewidth]{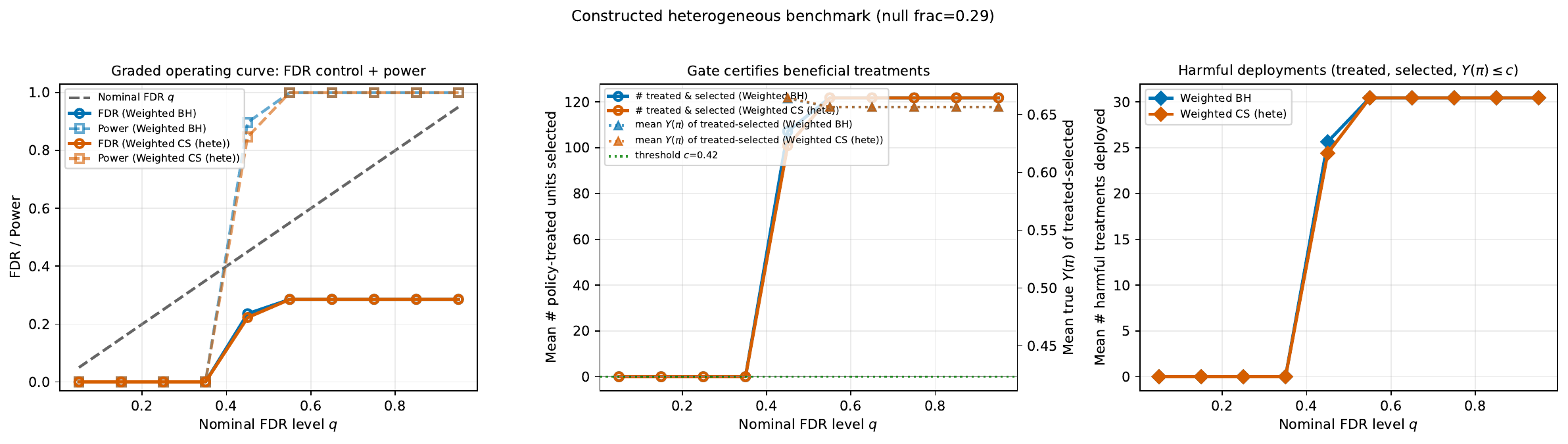}
\caption{Heterogeneous safety gate. Left, empirical FDR against the nominal level
with the diagonal for reference. Middle, selection power. Right, the counts of
beneficial treated units and harmful deployments as functions of $q$.}
\label{fig:si-hetero}
\end{figure}

\end{document}